\documentclass[12pt]{article}%
\usepackage{amsmath,amssymb,amsthm,amsfonts}
\usepackage{wasysym}
\usepackage{subcaption}
\usepackage{graphicx}
\usepackage{color}
\usepackage[colorlinks]{hyperref}
\usepackage{geometry}

\newtheorem{thm}{Theorem}
\newtheorem{cor}{Corollary}

\newcommand{\R}{\mathbb{R}}
\newcommand{\xb}{\hat{x}_*}
\newcommand{\wb}{\hat{w}_*}
\newcommand{\C}{\mathbb{C}}
\newcommand{\N}{\mathbb{N}}
\renewcommand{\d}{\hat{d}}
\newcommand{\lambar}{\bar{\lambda}}
\renewcommand{\Re}{\text{Re}}
\newcommand{\ea}{\textit{et al. }}

\begin{document}

\title{Neimark-Sacker bifurcations and evidence of chaos in a discrete dynamical model of walkers}
\author{Aminur Rahman, ar276@njit.edu, \url{web.njit.edu/~ar276}\\
Denis Blackmore, denis.l.blackmore@njit.edu\\ \\
{\small Department of Mathematical Sciences, New Jersey Institute of Technology; Newark, NJ 07102-1982}}
\date{}
\maketitle

\begin{abstract}
Bouncing droplets on a vibrating fluid bath can exhibit wave-particle behavior, such as being propelled by
interacting with its own wave field. These droplets seem to walk across the bath, and thus are dubbed
\emph{walkers}. Experiments have shown that walkers can exhibit exotic dynamical behavior indicative of chaos.
While the integro-differential models developed for these systems agree well with the experiments, they are 
difficult to analyze mathematically. In recent years, simpler discrete dynamical models have been derived and 
studied numerically. The numerical simulations of these models show evidence of exotic dynamics such as period 
doubling bifurcations, Neimark--Sacker (N--S) bifurcations, and even chaos. For example, in \cite{Gilet14}, based 
on simulations Gilet conjectured the existence of a supercritical N-S bifurcation as the damping factor in his one-
dimensional path model. We prove Gilet's conjecture and more; in fact, both supercritical and subcritical (N-S) 
bifurcations are produced by separately varying the damping factor and wave-particle coupling for all eigenmode
shapes. Then we compare our theoretical results with some previous and new numerical simulations, and find 
complete qualitative agreement. Furthermore, evidence of chaos is shown by numerically studying a global 
bifurcation.
\end{abstract}

Keywords:  hydrodynamic quantum analogs, bouncing droplets, chaos, bifurcations

Pacs numbers: 05.45.Ac, %Low-dimensional chaos %
%47.53.+n, %Fractals in fluid dynamics %
%05.45.Mt, %Quantum choas %05.45.Pq, %Numerical simulations of chaotic systems %
47.20.Ky, %Bifurcations in fluid dynamics %
47.52.+j %Chaos in fluid dynamics

\section{Introduction}

Pilot wave theory is based upon the idea that particle trajectories drive the statistics seen in quantum mechanics. 
In the early days of quantum mechanics it was seen as a promising explanation of the statistics arising in 
experiments. Specifically the works of de Broglie \cite{deBroglie} and later Bohm \cite{BohmVigier}, while very 
different, showed potential as alternate formulations to that of the Copenhagen interpretation until the 1960s. In 
recent years, fluid dynamic experiments with droplets bouncing on a vibrating bath were observed to exhibit 
quantum-like behavior. This effectively revived the theory in the form of hydrodynamic quantum analogs and 
catalyzed the current very active state of walking droplet research. A summary of quantum-like behavior in 
experiments of walking droplets can be found in \cite{Bush10, Bush15a, Bush15b}.

In 2005 Couder and collaborators observed droplets moving across a vibrating fluid bath \cite{CPFB05}. The 
droplets accomplish this by exciting an eigenmode at each bounce, which generates a wave field, and in turn are 
propelled by interacting with this field
\cite{CPFB05, PBC2006}. Later, Couder \ea\cite{CouderFort06} and
Eddi \ea\cite{EFMC09, EMPFC12} observed quantum-like behavior such as single particle diffraction, tunneling, and 
the Zeeman effect. Furthermore, there were extensive experiments conducted on the quantization of orbits by 
Fort \ea\cite{FEBMC10}, Harris \ea
\cite{HMFCB13, HarrisBush14}, and Perrard \ea\cite{PLMFC14, PLFC14}.

The experimental work led to many mathematical models. Discrete path memory models (accompanied by 
experimental results) for quantized circular orbits and walking in free space were developed by Fort \ea
\cite{FEBMC10} and Eddi \ea\cite{ESMFRC11}, respectively. A detailed hydrodynamic model derived from free 
space experiments was developed by Molacek \ea\cite{MolBush13b}. Oza \ea took this further and developed 
integro-differential equation models for walkers in free space \cite{ORB13} and a circular rotating frame
\cite{OHRB14, OWHRB14}. Then more accurate models for the waves were developed by Milewski \ea
\cite{MGNB15} and coupled with the equations of motion of the droplet.

While many of these models agree very well with experiments, the equations are quite difficult to study 
analytically. The complexity of the equations naturally created interest in developing realistic simplified 
mathematical models exhibiting important dynamical features of the original, while being easier to analyze. The 
first models to exploit the discrete nature of the bouncing/walking were developed by Fort \ea\cite{FEBMC10} 
and later Eddi \ea\cite{ESMFRC11}. Several other such reduced models were developed and two of them,
devised by Shirokoff \cite{Shirokoff13} and Gilet
\cite{Gilet14}, showing considerable promise are actually planar
discrete dynamical system models. Shirokoff \cite{Shirokoff13} developed a model in which he derived a map for 
the motion of a particle in a square cavity. In this model, using numerics, he discovered cascading period doubling 
bifurcations indicative of chaos. Gilet
\cite{Gilet14} included the amplitude of subsequent modes in his
model (\ref{Eq: themodel}) for the straight line motion of a particle. He observed what appeared to be a Neimark--
Sacker (N-S) bifurcation in numerous simulations, and so conjectured its existence and type.

\begin{equation}
\begin{split}
w_{n+1} &= \mu[w_n + \Psi(x_n)],\\
x_{n+1} &= x_n - Cw_n\Psi'(x_n);
\end{split}
\label{Eq: themodel}
\end{equation}

\begin{figure}
\centering
\includegraphics[width = 0.9\textwidth]{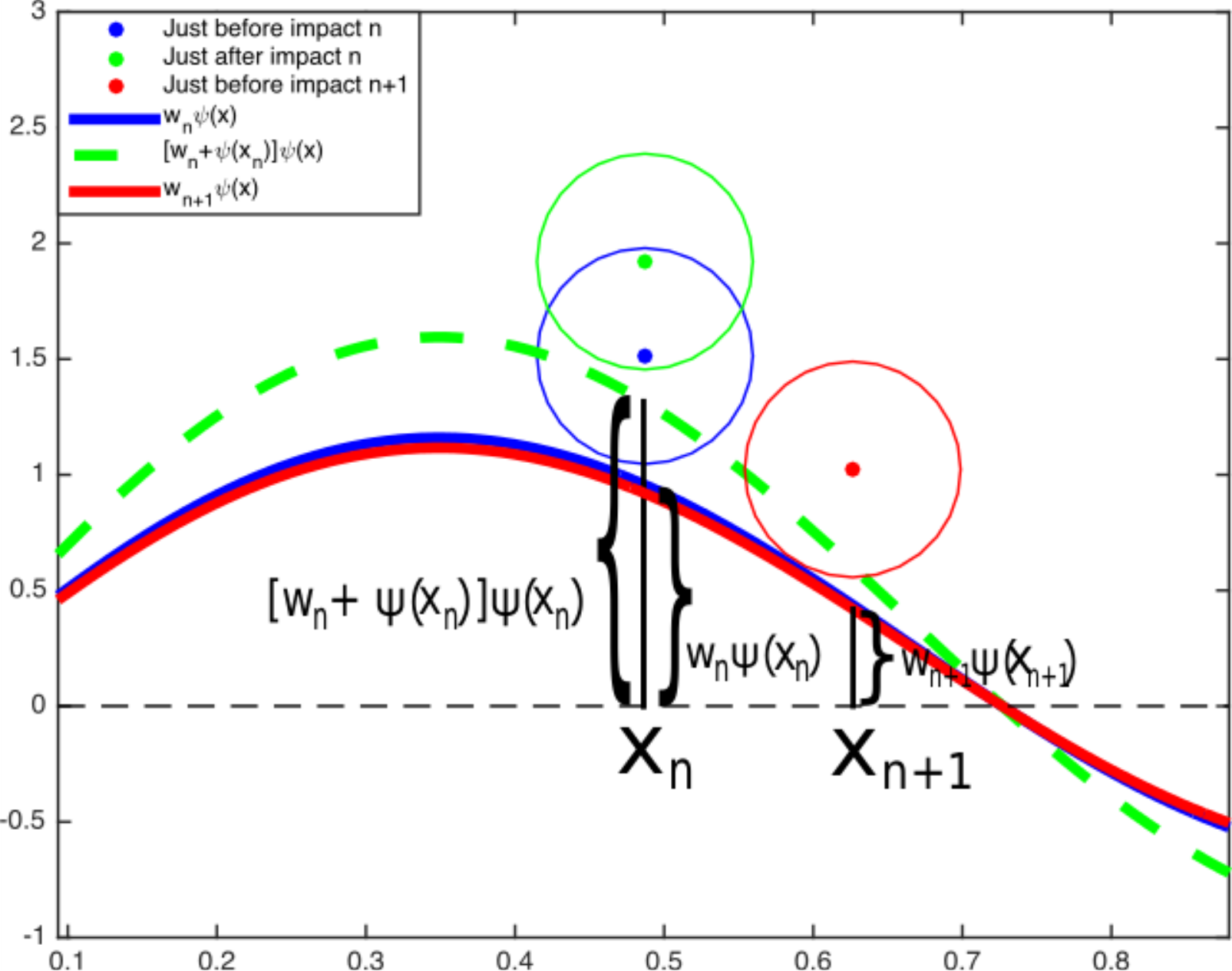}
\caption{Diagram of (\ref{Eq: themodel}).}
\label{Fig: Diagram}
\end{figure}

The model diagram is illustrated in Fig. \ref{Fig: Diagram}. In this model $w_n \in \R$ is the amplitude of the wave 
just before impact $n \in \N$, $x_n \in \R$ is the position of the walker at impact $n$, $C \in [0,1]$ represents 
wave-particle coupling (related to the size of the droplet), $\mu \in [0,1]$ is the damping factor (related to the 
path memory), and $\Psi \in \R$ is a single eigenmode of the Faraday wave field. From the definition given in 
\cite{Gilet14}, one can calculate the damping factor of mode $k$ with the associated memory $M_k$ as 
$\mu_k = \exp(-1/M_k)$. However, since this model has only one eigenmode, $k = 1$.  One may also interpret 
the damping factor as the ratio of the wave amplitude just before impact $n+1$ to the amplitude just after 
impact $n$.  Then, we think of $\mu = 0$ as a system with no memory and hence no walking.  Furthermore, $
\mu = 1$ corresponds to infinite memory and gives rise to Faraday waves; that is, there are waves present even 
when a drop is not.  In addition, $C = 0$ corresponds to a drop that is nonexistent and $C = 1$ represents a drop 
that is so large that the difference in position is exactly the gradient of the wave field at impact $n$.

While it is not possible to construct an exact experiment for this model because constraining the droplets to an 
annular or rectangular domain would cause boundary effects from the waves, we may think of this model as being 
analogous to or an approximation of some experimental setups. Two such experiments have been conducted by 
Filoux \ea \cite{FHV15, FHSV15}; the first being walkers in an annular domain and the second an upcoming work 
on walkers in a long-narrow rectangular domain.

Although the focus of this investigation is rigorous mathematical analysis of certain aspects of the model 
dynamics, in order to connect the theory with the physics we shall discuss the physical implications of the 
dynamical results. Furthermore, the theory developed will give us insight into experiments such as
\cite{FHV15, FHSV15} and perhaps even 2-dimensional free space
experiments due to the connections between discrete dynamical systems and continuous dynamical systems.

The remainder of this paper is organized as follows: In section
\ref{Sec: Basic} we prove the basic dynamical properties of
(\ref{Eq: themodel}). We analyze the map via dynamical systems and bifurcation theory in section \ref{Sec: NS} 
and formulate the conditions for it to have N--S bifurcations. In our analysis, we first vary $\mu$ while treating
$C$ as a constant, then vice versa, and find the value of $\mu$ and $C$ for which the eigenvalues have a 
modulus of unity and become complex conjugates. Then, we apply the genericity conditions for a Neimark--Sacker 
bifurcation to determine conditions on $\Psi$ which make the system a generic 2-D map. In section
\ref{Sec:  Test} we apply our theorem for the conditions on $\Psi$ to test functions that Gilet proposed. 
While we use the same form of test functions, we also study test function shapes different from those 
studied in \cite{Gilet14}. We end in section \ref{Sec: Homoclinic} by noting
an interesting homoclinic-like bifurcation that appears in our simulations.

\section{Basic properties of the map}\label{Sec: Basic}

In this section we find the fixed points as done in \cite{Gilet14}. Then we analyze the stability of the fixed points 
for two cases: holding $C$ constant and holding $\mu$ constant. It should be noted
that a family of fixed points is defined as a set of multiple fixed points,
and all calculations are done without loss of generality.

First, define $F: \R^2 \mapsto \R^2$ as,

\begin{equation}
\left[\begin{array}{c}
w_{n+1}\\
x_{n+1}
\end{array}\right] = F(w_n,x_n) = \left[\begin{array}{cc}
\mu[w_n + \Psi(x_n)]\\
x_n - Cw_n\Psi'(x_n)
\end{array}\right]\label{Eq: themap}
\end{equation}

We find the fixed points of (\ref{Eq: themodel}) by solving,

\begin{align*}
w_* &= \mu[w_* + \Psi(x_*)]\\
x_* &= x_* - Cw_*\Psi'(x_*).
\end{align*}

Notice that this gives us two families of fixed points:

\begin{align}
(0,x_*) \text{ such that } &\Psi(x_*) = 0;\; \mu \neq 0\, \text{ and}\label{Eq: fp1}\\ (\wb,\xb) 
\text{ such that } &\wb = \frac{\mu}{1-\mu}\Psi(\xb),
\Psi'(\xb) = 0;\; C \neq 0, \mu \neq 1.\label{Eq: fp2}
\end{align}

The derivative matrix for $F$ is,

\begin{equation}
DF(w,x) = \left(\begin{array}{cc}
\mu & \mu\Psi'(x)\\
-C\Psi'(x) & 1-Cw\Psi''(x)
\end{array}\right).\label{Eq: Jacobian}
\end{equation}

\subsection{Properties of $F$ in neighborhoods of the second family}

We first notice that for fixed points (\ref{Eq: fp2}), $\mu \in [0,1)$ and $C \in (0,1]$ because if $\mu = 1$ or 
$C = 0$ this fixed point does not exist. Substituting the fixed points (\ref{Eq: fp2}) into the derivative matrix 
(\ref{Eq: Jacobian}), we obtain

\begin{equation*}
DF(\wb,\xb) = \left(\begin{array}{cc}
\mu & 0\\
0 & 1-\frac{\mu C}{1-\mu}\Psi(\xb)\Psi''(\xb)
\end{array}\right).
\end{equation*}

Therefore, the eigenvalues are

\begin{equation*}
\lambda_1 = \mu,\, \lambda_2 = 1-\frac{\mu C}{1-\mu}\Psi(\xb)\Psi''(\xb)
\end{equation*}

Observe that since $\mu \in [0,1)$, if $\Psi(\xb)\Psi''(\xb) < 0$ the fixed point is a saddle.  However, if 
$\Psi(\xb)\Psi''(\xb) > 0$, we always have $\lambda_2 < 1$, so we must find a condition for which
$\lambda_2 = -1$.

If we hold $C$ constant we derive the condition $\hat{\mu}$, 
\begin{equation*}
\hat{\mu} = \frac{1}{1+C\Psi(\xb)\Psi''(\xb)/2}.
\end{equation*}
Then, for $\mu < \hat{\mu}$ this fixed point is a sink and for $\mu > \hat{\mu}$ this fixed point is a saddle.
If $\mu = \hat{\mu}$, the fixed point becomes nonhyperbolic.

If we hold $\mu$ constant we derive the condition $\hat{C}$, 
\begin{equation}
\hat{C} = \frac{2(1-\mu)}{\mu\Psi(\xb)\Psi''(\xb)}.
\end{equation}
If $C < \hat{C}$ the fixed points are stable, but if $C > \hat{C}$ the fixed points are saddles.
When $C = \hat{C}$, $\lambda_2 = -1$, which indicates the possibility of a flip bifurcation.

When the fixed point is a saddle, it is easy to see that the linear stable and unstable manifolds are, 
\begin{align}
W^s_{\text{lin}} &= \{(w,x): x = \xb\},\\ W^u_{\text{lin}} &= \{(w,x): w = \wb \frac{\mu}{1-\mu}\Psi(\xb)\}.
\end{align}

\subsection{Properties of $F$ in neighborhoods of the first family}\label{Sec: fp1}

Since we can translate any fixed point (\ref{Eq: fp1}) to the origin, without loss of generality, we assume the fixed 
point is $(w_*,x_*) = (0,0)$ with the relevant conditions on $\Psi'$. Substituting the fixed point (\ref{Eq: fp1}) 
into the derivative matrix (\ref{Eq: Jacobian}), we obtain
 
\begin{equation}
DF(0,0) = \left(\begin{array}{cc}
\mu & \mu\Psi'(0)\\
-C\Psi'(0) & 1
\end{array}\right).
\end{equation}

The characteristic polynomial is 
\begin{equation*}
(\mu-\lambda)(1-\lambda)+C\mu\Psi'(0)^2 = 0 \Rightarrow
\lambda^2 - (1+\mu)\lambda + \mu(1+C\Psi'(0)^2) = 0,
\end{equation*}
and the eigenvalues are,
\begin{equation}
\lambda = \frac{1}{2}(1+\mu) \pm \frac{i}{2}\sqrt{4\mu(1+C\Psi'(0)^2)-(1+\mu)^2}.
\end{equation}
Notice that $\lambda \in \C\setminus\R$ if $4\mu(1+C\Psi'(0)^2)-(1+\mu)^2>0$.

Holding $C$ constant gives us the condition, 
\begin{equation}
2C\Psi'(0)^2 - 2\sqrt{C\Psi'(0)^2(C\Psi'(0)^2+1)} + 1< \mu < 2C\Psi'(0)^2 + 2\sqrt{C\Psi'(0)^2(C
\Psi'(0)^2+1)} + 1
\end{equation}
Now, if $|\lambda|<1$, we get a stable focus, and when $|\lambda|>1$ we get an unstable focus. This indicates 
that there may be a N--S bifurcation at the fixed point when $|\lambda|$ passes through unity, which occurs when 
$\mu$ goes from $\mu \leq \mu_* : = 1/(1+C\Psi'(0)^2)$ to $\mu > \mu_*$.

If we hold $\mu$ constant we have, 
\begin{equation}
C > \frac{1}{\Psi'(0)^2}\left[\frac{(1+\mu)^2}{4\mu} - 1\right],
\end{equation}
and plugging in this condition shows that the fixed points are always stable when the eigenvalue is real. When the 
eigenvalues are purely complex conjugates the fixed points undergo N--S bifurcations, similar to the constant $C$ 
case. The eigenvalue passes through unity when $C_* = (1/\mu - 1)/\Psi'(0)^2$.

\section{Neimark--Sacker bifurcation}\label{Sec: NS}

In a continuous two-dimensional dynamical system, a Hopf bifurcation occurs when a spiral fixed point changes 
stability and produces a limit cycle as a parameter is varied. The bifurcation is supercritical if a stable limit cycle 
arises as the parameter is varied forward. Conversely, the bifurcation is subcritical if an unstable limit cycle is 
created as we vary the parameter backward. Similarly, for a discrete two-dimensional map a N--S bifurcation 
occurs in the same manner. Due to the discrete nature of the dynamical system, an invariant closed Jordan curve 
(a topological circle) is born instead of a limit cycle.

The existence of a N--S bifurcation at a specific fixed point for a single parameter can be shown by verifying the 
relevant genericity conditions and showing that the eigenvalues are complex conjugates of modulus one as the 
parameter passes through the critical value. The genericity condition is pivotal in proving a bifurcation (a 
topological change in the set of iterates) exists. Furthermore, the complex conjugate eigenvalues prove that the 
fixed point is a foci before and after the bifurcation, thereby allowing the creation of an invariant circle.

We consider the genericity conditions outlined in \cite{Kuznetsov}.  That is, the map must be locally conjugate 
near the fixed point to a specific normal form, there are no strong resonances, and the first Lyapunov coefficient 
must be nonzero.  The Lyapunov coefficient also determines if the bifurcation is supercritical (if negative) or 
subcritical (if positive).

While the N--S bifurcation is a local phenomenon about a single fixed point, it should be noted that all calculations 
are done without loss of generality. Therefore, we can apply the theorem to individual fixed points, some of which 
will yield N--S bifurcations.

\subsection{Neimark--Sacker bifurcation in the parameter $\mu$}\label{Sec: NS-mu}

In \cite{Gilet14}, Gilet conjectured that a supercritical N--S bifurcation occurs at the fixed points (\ref{Eq: fp1}). 
He also observed evidence of this in the iterates of the map for chosen test functions. Here we prove the map 
(\ref{Eq: themap}) is generic and a N--S bifurcation occurs as we vary the parameter $\mu$ at the fixed point 
(c.f. \cite{Kuznetsov, Neimark, Sacker}). We also show the map allows for both supercritical and subcritical N--S 
bifurcations.

\begin{thm}
The map (\ref{Eq: themap}) is generic about some fixed point $(w_*,x_*)$ if the eigenmode satisfies the 
following property, 
\begin{align}
\d = &\Psi'''(x_*)\Psi'(x_*)(1+C\Psi'(x_*)^2)(1
+2C\Psi'(x_*)^2)(4+3C\Psi'(x_*)^2)\nonumber\\ &+ 2\Psi''(x_*)^2\left\{5+C\Psi'(x_*)^2\left[1\right.\right.
\left.\left.-C\Psi'(x_*)^2\left(31+21C\Psi'(x_*)^2\right)\right]\right\} \neq 0.
\end{align}
and a Neimark--Sacker bifurcation occurs at the fixed points (\ref{Eq: fp1}) when 
\begin{equation}
\mu = \mu_* = \frac{1}{1+C\Psi'(x_*)^2}.
\end{equation}
Furthermore, if $\d < 0$, the map undergoes a supercritical Neimark--Sacker bifurcation, and if $\d > 0$, the map 
undergoes a subcritical Neimark--Sacker bifurcation in a neighborhood of the fixed point.
\label{Thm: thetheorem}
\end{thm}

\begin{proof}
As in Section \ref{Sec: fp1}, without loss of generality we translate the fixed point to $(w_*,x_*)=(0,0)$.  We 
showed in Section \ref{Sec: fp1} that the pair of eigenvalues $\lambda$ are complex conjugates if 
\begin{equation}
2C\Psi'(0)^2 - 2\sqrt{C\Psi'(0)^2(C\Psi'(0)^2+1)} + 1< \mu 
< 2C\Psi'(0)^2 + 2\sqrt{C\Psi'(0)^2(C\Psi'(0)^2+1)} + 1
\end{equation}
and $|\lambda|=1$ when 
\begin{equation*}
\mu = \mu_* = \frac{1}{1+C\Psi'(0)^2}.
\end{equation*}
This shows that a Neimark--Sacker bifurcation occurs at the fixed point if the map is generic.

Next we show the map is generic (c.f. \cite{Kuznetsov, Neimark, Sacker})
via three conditions (C.1), (C.2), and (C.3),

\begin{itemize}

\item[(C.1)]  We show that $r'(\mu_*) \neq 0$, where $r = |\lambda| = \sqrt{\mu\left(1+C\Psi'(0)^2\right)}$

Notice, since $r(\mu_*) = 1$,

\begin{equation*}
\frac{d}{d\mu}\left(r(\mu)^2\right)\bigg|_{\mu = \mu_*} = 2r(\mu_*)r'(\mu_*) = 2r'(\mu_*),
\end{equation*}

so if $\frac{d}{d\mu}\left(r(\mu)^2\right)\bigg|_{\mu = \mu_*} \neq 0$, $r'(\mu_*) \neq 0$.
Then, since $C > 0$ and $\Psi'(0) \in \R$,

\begin{equation}
\frac{d}{d\mu}\left(r(\mu)^2\right)\bigg|_{\mu = \mu_*} = (1+C\Psi'(0)^2) \neq 0.
\end{equation}

This shows that the transversality condition is satisfied.

\item[(C.2)]  We show the arguments of the eigenvalues satisfy the first set of nondegeneracy
conditions for a N--S bifurcation.

Let $\theta_* = \tan^{-1}A$, where

\begin{equation}
A = \frac{\sqrt{4\mu_*(1+C\Psi'(0)^2)-(1+\mu_*)^2}}{1+\mu_*} = \frac{\sqrt{4-(1+\mu_*)^2}}{1+\mu_*}
\end{equation}

Observe that $\theta_* = 0$ if $A = 0$, $\theta_* = \pm \pi$ if $A = 0$, $\theta_* = \pm 2\pi/3$
if $A = \pm \sqrt{3}$, and $\theta_* \rightarrow \pm \pi/2$ as $A \rightarrow \pm\infty$.
Since $A$ is clearly positive and bounded, this rules out each case except $A = \sqrt{3}$.
In order to get $\sqrt{3}$ we need $1+\mu = 1$; however, since $\mu$ is positive this is not possible.

Thus the first nondegeneracy condition is satisfied.

\item[(C.3)]  We compute the normal form for the N--S bifurcation and derive
the conditions for which the second nondegeneracy condition is satisfied.

First, we compute the eigenvectors: $(DF)q = \lambda q$ and $(DF)^Tp = \lambar p$,

\begin{equation}
q = \left(\begin{array}{c}
\mu_*\Psi'(0)\\
\lambda - \mu_*
\end{array}\right) \text{ and } p = \left(\begin{array}{c}
\lambar-1\\
\mu_*\Psi'(0)
\end{array}\right)
\end{equation}

For the normalization, we take

\begin{equation}
<p,q> = \bar{p}\cdot q = i\mu_*\Psi'(0)\sqrt{4-(1+\mu_*)^2}.
\end{equation}

In order to simplify the calculations for the normal form let us change variables from $\R^2$ to $\C$

\begin{equation}
\left(\begin{array}{c}
w\\ x
\end{array}\right) = zq + \bar{z}\bar{q} = \left[\begin{array}{c}
(z+\bar{z})\mu_*\Psi'(0)\\ (\lambda - \mu_*)z + (\lambar - \mu_*)\bar{z}
\end{array}\right]
\end{equation}

Substituting this into $F$ yields

\begin{align}
F = &\left[\begin{array}{c} F_1\\ F_2
\end{array}\right];\;\text{where}\\
F_1 = &(z + \bar{z})\mu_*^2\Psi'(0) + \mu_*\Psi((\lambda - \mu_*)z
+ (\lambar - \mu_*)\bar{z}),\nonumber\\ F_2 = &(\lambda - \mu_*)z + (\lambar - \mu_*)\bar{z}
- C(z+\bar{z})\mu_*\Psi'(0)\Psi'((\lambda-\mu_*)z+(\lambar-\mu_*)\bar{z}).\nonumber
\end{align}

Now we take the inner product (in order to transform the map into the normal form)

\begin{align}
<p,F> = &(z+\bar{z})\mu_*^2(\lambda - 1)\Psi'(0) + \mu_*(\lambda - 1)\Psi((\lambda-\mu_*)z
+(\lambar-\mu_*)\bar{z}) + (\lambda-\mu_*)\mu_*\Psi'(0)z\nonumber\\
&+ (\lambar-\mu_*)\mu_*\Psi'(0)\bar{z}
- C\mu_*^2\Psi'(0)^2(z+\bar{z})\Psi'((\lambda-\mu_*)z+(\lambar-\mu_*)\bar{z}).
\end{align}

Finally, to get the normal form we divide this through by $<p,q>$ for the sake of normalization and take the 
Taylor series in order to ignore higher order terms,

\begin{align}
H = &\left[(z+\bar{z})\mu_*^2(\lambda - 1)\Psi'(0) + \mu_*(\lambda - 1)\sum_{j+k\geq 1}\frac{1}{j!k!}
\partial_{z^j\bar{z}^k}\Psi(0)z^j\bar{z}^k\right.
+ (\lambda-\mu_*)\mu_*\Psi'(0)z \nonumber\\ &+ (\lambar-\mu_*)\mu_*\Psi'(0)\bar{z}
\left. -C\mu_*^2\Psi'(0)^2(z+\bar{z})\left(\Psi'(0) + \sum_{j+k\geq 1}\frac{1}{j!k!}\partial_{z^j\bar{z}^k}
\Psi'(0)z^j\bar{z}^k\right)\right]\nonumber\\
&\bigg/\left(i\mu_*\Psi'(0)\sqrt{4-(1+\mu_*)^2}\right),
\end{align}

where $\partial_{z^j\bar{z}^k}\Psi(0) := \partial_{z^j\bar{z}^k}\Psi((\lambda-\mu_*)z
+(\lambar-\mu_*)\bar{z})|_{(z,\bar{z})=(0,0)}$, and similarly for $\partial_{z^j\bar{z}^k}\Psi'(0)$.

By matching linear terms, it is easy to show the normal form can be written as,

\begin{align}
H = \lambda z + &\left[(\lambda - 1)\sum_{j+k\geq 2}\frac{1}{j!k!}
\partial_{z^j\bar{z}^k}\Psi(0)z^j\bar{z}^k
- C\mu_*\Psi'(0)^2(z+\bar{z})\sum_{j+k\geq 1}\frac{1}{j!k!}\partial_{z^j\bar{z}^k}
\Psi'(0)z^j\bar{z}^k\right]\nonumber\\
&\bigg/\left(i\Psi'(0)\sqrt{4-(1+\mu_*)^2}\right).
\end{align}

We are ready now to compute the nondegeneracy condition required to satisfy the final genericity condition.
From \cite{Kuznetsov, Sacker}, we have the formula,

\begin{equation}
d(0) = \Re\left(\frac{\lambar g_{21}}{2}\right)
+ \Re\left(\frac{\lambar(\lambar-2)}{2(\lambda-1)}g_{20}g_{11}\right)
- \frac{1}{2}|g_{11}|^2 - \frac{1}{4}|g_{02}|^2.
\end{equation}

where $g_{jk}/(j!k!)$ is the coefficient of the $z^j\bar{z}^k$ term.  We compute the relevant terms

\begin{align*}
g_{20} = &\frac{(\lambda - 1)(\lambda - \mu_*)^2\Psi''(0) - 2C\mu_*\Psi'(0)^2(\lambda - \mu_*)\Psi'(0)} {i\sqrt{4-(1+\mu_*)^2}\Psi'(0)}\\ g_{02} = &\frac{(\lambda - 1)(\lambar - \mu_*)^2\Psi''(0) - 2C\mu_*\Psi'(0)^2(\lambar - \mu_*)\Psi'(0)} {i\sqrt{4-(1+\mu_*)^2}\Psi'(0)}\\ g_{11} = &\frac{(\lambda - 1)(\lambda - \mu_*)(\lambar-\mu_*)\Psi''(0)}{i\sqrt{4-(1+\mu_*)^2}\Psi'(0)} - \frac{C\mu_*\Psi'(0)^2(\lambda - \mu_* + \lambar - \mu_*)\Psi'(0)}{i\sqrt{4-(1+\mu_*)^2}\Psi'(0)}\\ g_{21} = &\frac{(\lambda - 1)(\lambda - \mu_*)^2(\lambar-\mu_*)\Psi'''(0)}{i\sqrt{4-(1+\mu_*)^2}\Psi'(0)} - \frac{C\mu_*\Psi'(0)^2((\lambda - \mu_*)^2 + 2(\lambda-\mu_*)(\lambar - \mu_*))\Psi'''(0)}{i\sqrt{4-(1+\mu_*)^2}\Psi'(0)}
\end{align*}

If we factor out certain terms this formula greatly simplifies to

\begin{equation}
d(0) = \frac{\Psi'''(0)}{2\Psi'(0)}\frac{\Re(-i\lambar\hat{g}_{21})}{\sqrt{4-(1+\mu_*)^2}} -\frac{\Psi''(0)^2}{\Psi'(0)^2}\frac{\Re\left(\frac{\lambar(\lambar-2)}{2(\lambda-1)}
\hat{g}_{20}\hat{g}_{11}\right)}{(4-(1+\mu_*)^2)}
- \frac{\Psi''(0)^2}{\Psi'(0)^2}\frac{\frac{1}{2}|\hat{g}_{11}|^2 + \frac{1}{4}|\hat{g}_{02}|^2}{(4-(1+\mu_*)^2)},
\end{equation}

where

\begin{align}
\hat{g}_{20} = (\lambda - \mu_*)&\left[\lambda^2 + \mu_* - 2C\Psi'(0)^2\mu_*- \lambda(1+\mu_*)\right]\\
\hat{g}_{02} = (\lambar - \mu_*)&\left[1+\mu_* - 2C\Psi'(0)^2\mu_* - \lambda\mu_* - \lambar\right]\\
\hat{g}_{11} = -1+\lambda +&\lambda\mu_* - C\Psi'(0)^2\lambda\mu_* - \lambda^2\mu_*
+ \lambda\mu_* - C\Psi'(0)^2\lambar\mu_*\nonumber\\
&- \mu_* - \mu_*^2 + 2C\Psi'(0)^2\mu_*^2 + \lambda\mu_*^2\\
\hat{g}_{21} = (\mu_* - \lambda)&\left[1-\lambda - \lambda\mu_* + C\Psi'(0)^2\lambda\mu_*
+\lambda^2\mu_* - \lambar\mu_* + 2C\Psi'(0)^2\lambar\mu_*\right.\nonumber\\
&\left. + \mu_* + \mu_*^2 -3C\Psi'(0)^2\mu_*^2 - \lambda\mu_*^2\right].
\label{Eq: SameArgument}
\end{align}

Then, substituting in for $\lambda$ and $\mu_*$ gives

\begin{align}
d(0) = C^2\Psi'(0)^2\bigg[&\Psi'''(0)\Psi'(0)(1+C\Psi'(0)^2)(1+2C\Psi'(0)^2)(4+3C\Psi'(0)^2)\nonumber\\
&+ 2\Psi''(0)^2\left\{5 +C\Psi'(0)^2\left[1 C\Psi'(0)^2\left(31+21C\Psi'(0)^2\right)\right]\right\}\bigg]\nonumber\\ &\bigg/\left[4(1+ C\Psi'(0)^2)^4(4+3C\Psi'(0)^2)\right].
\end{align}

Since

\begin{equation*}
\frac{C^2\Psi'(0)^2}{4(1+C\Psi'(0)^2)^4(4+3C\Psi'(0)^2)} > 0,
\end{equation*}

we need only be concerned with

\begin{align}
\d := &\Psi'''(x_*)\Psi'(x_*)(1+C\Psi'(x_*)^2)(1+2C\Psi'(x_*)^2)(4+3C\Psi'(x_*)^2)\nonumber\\
&+ 2\Psi''(x_*)^2\left\{5+C\Psi'(x_*)^2\left[1-C\Psi'(x_*)^2\left(31+21C\Psi'(x_*)^2\right)\right]\right\}
\end{align}

Consequently, the N--S criteria imply that the bifurcation occurs if $\d \neq 0$,
and it is supercritical for $\d < 0$ and subcritical for $\d > 0$.

This shows that for certain properties of $\Psi$ the second nondegeneracy condition is satisfied,
thereby completing the proof.

\end{itemize}

\end{proof}

\subsection{Neimark--Sacker bifurcation in the parameter $C$}\label{Sec: NS-C}

In this section we prove the map (\ref{Eq: themap}) is generic and a N--S bifurcation
occurs as we vary the parameter $C$ at the fixed point.

\begin{cor}
The map (\ref{Eq: themap}) is generic about some fixed point $(w_*,x_*)$
if the eigenmode satisfies the following property,
\begin{align}
\d = \Psi'(x_*)\Psi'''(x_*)\left(6-\mu-\mu^2\right) + 2\Psi''(x_*)^2\left(32\mu-6\mu^3-21\right) \neq 0.
\end{align}
and a N--S bifurcation occurs at the fixed points (\ref{Eq: fp1}) when
\begin{equation}
C = C_* = \frac{1}{\Psi'(x_*)^2}\left(\frac{1}{\mu}-1\right).
\end{equation}
Furthermore, if $\d < 0$, the map undergoes a supercritical N--S bifurcation, and if $\d > 0$,
the map undergoes a subcritical N--S bifurcation in a neighborhood of the fixed point.
\label{Thm: cor}
\end{cor}

\begin{proof}

As in Section \ref{Sec: fp1}, without loss of generality we translate the fixed point to $(w_*,x_*)=(0,0)$.
We showed in Section \ref{Sec: fp1} that the pair of eigenvalues $\lambda$ are complex conjugates if 
\begin{equation}
C > \frac{1}{\Psi'(0)^2}\left[\frac{(1+\mu)^2}{4\mu} - 1\right],
\end{equation}
 and $|\lambda|=1$ when 
\begin{equation*}
C = C_* = \frac{1}{\Psi'(0)^2}\left(\frac{1}{\mu}-1\right).
\end{equation*}
This shows that a Neimark--Sacker bifurcation occurs at the fixed point if the map is generic.

Most of the arguments for genericity follow directly from the proof of Theorem \ref{Thm: thetheorem}.
For the first condition we use the same $r$ and study $r(C)$.  Just as with $r(\mu)$,
it suffices to study $r(C)^2$.  Then, 
\begin{equation}
\frac{d}{dC}\left(r(C)^2\right)\bigg|_{C=C_*} = \mu\Psi'(0)^2 \neq 0 \text{ if } \Psi'(0) \neq 0.
\end{equation}
Hence, if $\Psi'(0) \neq 0$, the map satisfies the transversality condition.

The argument to show no strong resonances is the same as that of Theorem \ref{Thm: thetheorem}.

For the first Lyapunov coefficient, the calculations are the same up to (\ref{Eq: SameArgument}),
by replacing $C$ with $C_*$ and $\mu_*$ with $\mu$.  Then, substituting in for $\lambda$ and $C_*$ gives 
\begin{equation}
d(0) = \frac{(\mu-1)^2}{4\Psi'(0)^2(3+\mu)}\left[\Psi'(x_*)\Psi'''(x_*)\left(6-\mu-\mu^2\right)
+ 2\Psi''(x_*)^2\left(32\mu-6\mu^3-21\right)\right].
\end{equation}
Since 
\begin{equation*}
\frac{(\mu-1)^2}{4\Psi'(0)^2(3+\mu)} > 0
\end{equation*}
we need only be concerned with, 
\begin{equation}
\d := \Psi'(x_*)\Psi'''(x_*)\left(6-\mu-\mu^2\right) + 2\Psi''(x_*)^2\left(32\mu-6\mu^3-21\right)
\end{equation}
Consequently, the Neimark--Sacker criteria imply that the bifurcation occurs if $\d \neq 0$,
and it is supercritical for $\d < 0$ and subcritical for $\d > 0$.

This shows that for certain properties of $\Psi$ the genericity conditions are satisfied,
thereby completing the proof.

\end{proof}

\section{Application of theoretical results to test functions}\label{Sec: Test}

In \cite{Gilet14}, Gilet uses the test functions, 
\begin{equation}
\Psi(x,\beta) = \frac{1}{\sqrt{\pi}}\cos\beta\sin3x + \frac{1}{\sqrt{\pi}}\sin\beta\sin5x
\end{equation}
where $\beta \in [0,\pi]$ is a fixed parameter that can be changed in order to tweak the shape of the eigenmode. 
He shows numerical constructions of the iterate-space for $\beta = \pi/3$ and $\beta = \pi/6$, and studies the 
statistics of the iterates for other values of $\beta$, all the while holding $C = 0.05$. We first follow in his 
footsteps and hold $C = 0.05$ and vary $\mu$, then we hold $\mu = 0.5$ and vary $C$.

To illustrate how small the parameter regime for a subcritical N--S bifurcation is, let us consider $\d$ for the 
origin, since it is always a fixed point.  We plot over $\beta \in [0,\pi]$, first for $\mu_*$ (i.e. holding
$C = 0.05$), then for $C_*$ (i.e. holding $\mu = 0.5$). This is shown in Fig. \ref{Fig: d}.  Recall that
$\d < 0$ denotes a supercritical N--S and $\d > 0$ denotes subcritical.

\begin{figure}[htbp]
\centering
\begin{subfigure}[t]{0.49\textwidth}
\includegraphics[width=\textwidth]{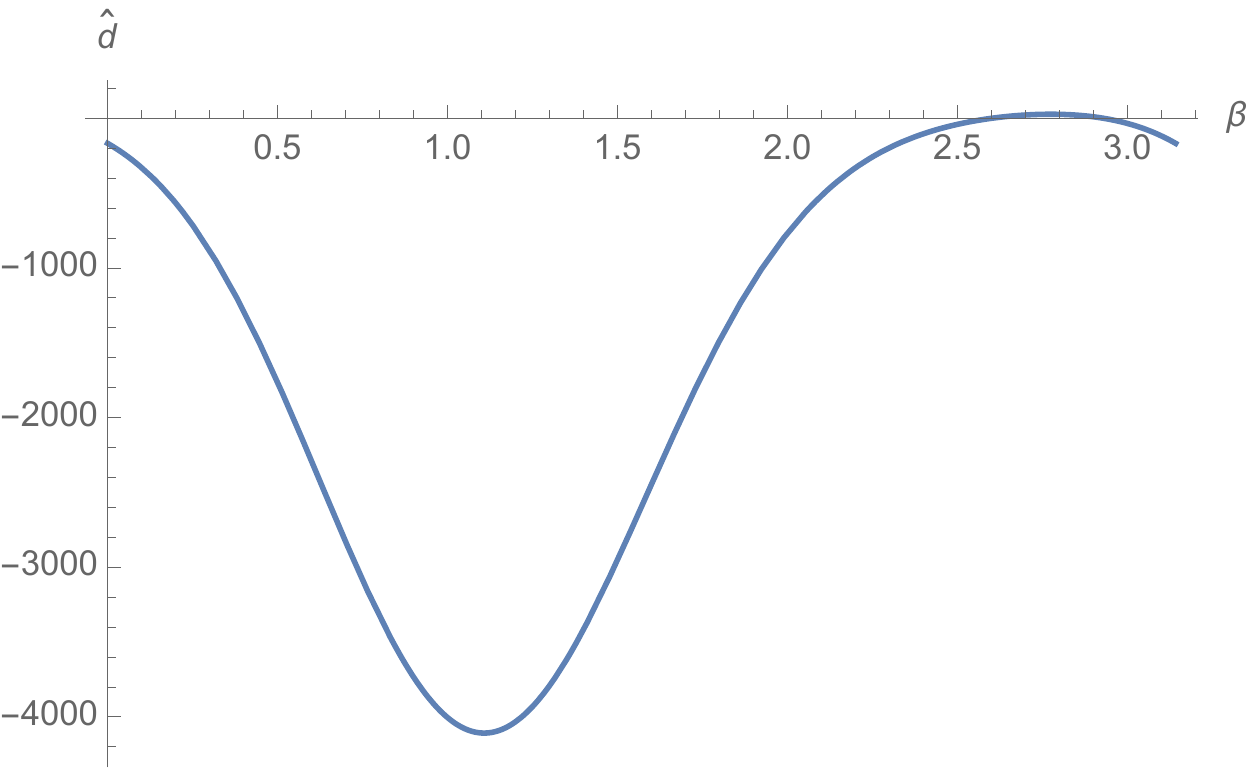}
\caption{$\mu = \mu_*$, $C = 0.05$, $(x_*,w_*) = (0,0)$}
\label{Fig: d_mu}
\end{subfigure}
\begin{subfigure}[t]{0.49\textwidth}
\includegraphics[width=\textwidth]{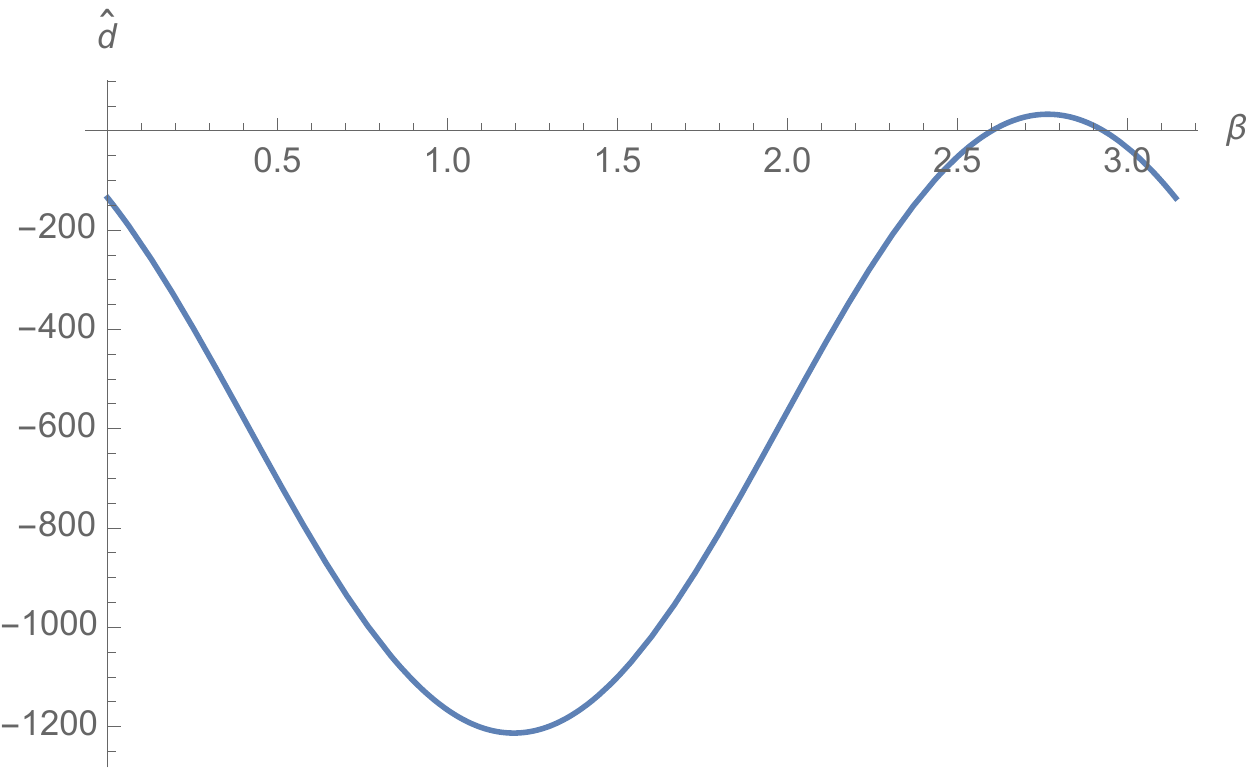}
\caption{$\mu = 0.5$, $C = C_*$, $(x_*,w_*) = (0,0)$}
\label{Fig: d_C}
\end{subfigure}
\caption{Plot of first Lyapunov coefficient for Neimark--Sacker bifurcations in $\mu$
and $C$ respectively.}
\label{Fig: d}
\end{figure}

\subsection{Simulation results for Neimark--Sacker bifurcation in $\mu$}\label{Sec: Bif_mu}

As in \cite{Gilet14}, we first study the eigenmode for $\beta = \pi/3$, for which the map exhibits supercritical
N--S bifurcations at various fixed points, but we also study $\beta = 5\pi/6$, for which the map exhibits a 
subcritical N--S bifurcation at the origin.  We hold $C = 0.05$, which corresponds to a droplet much smaller than 
the cavity being created by an impact, and we vary $\mu$, which corresponds to the system approaching high 
memory.

It should be noted that for the sake of plotting, we consider the map (\ref{Eq: themap}) to be written as 
\begin{equation}
\hspace{-0.1cm}\left[\begin{array}{c}
x_{n+1}\\ w_{n+1}
\end{array}\right] = F(x_n,w_n) = \left[\begin{array}{c}
x_n - Cw_n\Psi'(x_n)\\
\mu[w_n + \Psi(x_n)]
\end{array}\right]
\label{Eq: theothermap}
\end{equation}

\subsubsection{Supercritical Neimark--Sacker bifurcations for $\beta = \pi/3$}

For $\beta = \pi/3$, we now consider the map (\ref{Eq: theothermap}) on the domain $x\in[0,\pi/2]$.  On this 
domain, the map has three fixed points $(x_*,0)$ about which N--S bifurcations may occur; one of which is the 
origin with the other two being the following:

\begin{align*}
x_* &= \tan^{-1}\left(\sqrt{\frac{15+\sqrt{3}-\sqrt{6(8-\sqrt{3})}}{9-\sqrt{3}+\sqrt{6(8-\sqrt{3})}}}\right)
\approx 0.7269\\
x_* &= \tan^{-1}\left(\sqrt{\frac{15+\sqrt{3}+\sqrt{6(8-\sqrt{3})}}{9-\sqrt{3}-\sqrt{6(8-\sqrt{3})}}}\right)
\approx 1.3515
\end{align*}

\begin{figure}[htbp]
\centering
\begin{subfigure}[t]{0.49\textwidth}
\includegraphics[width=\textwidth]{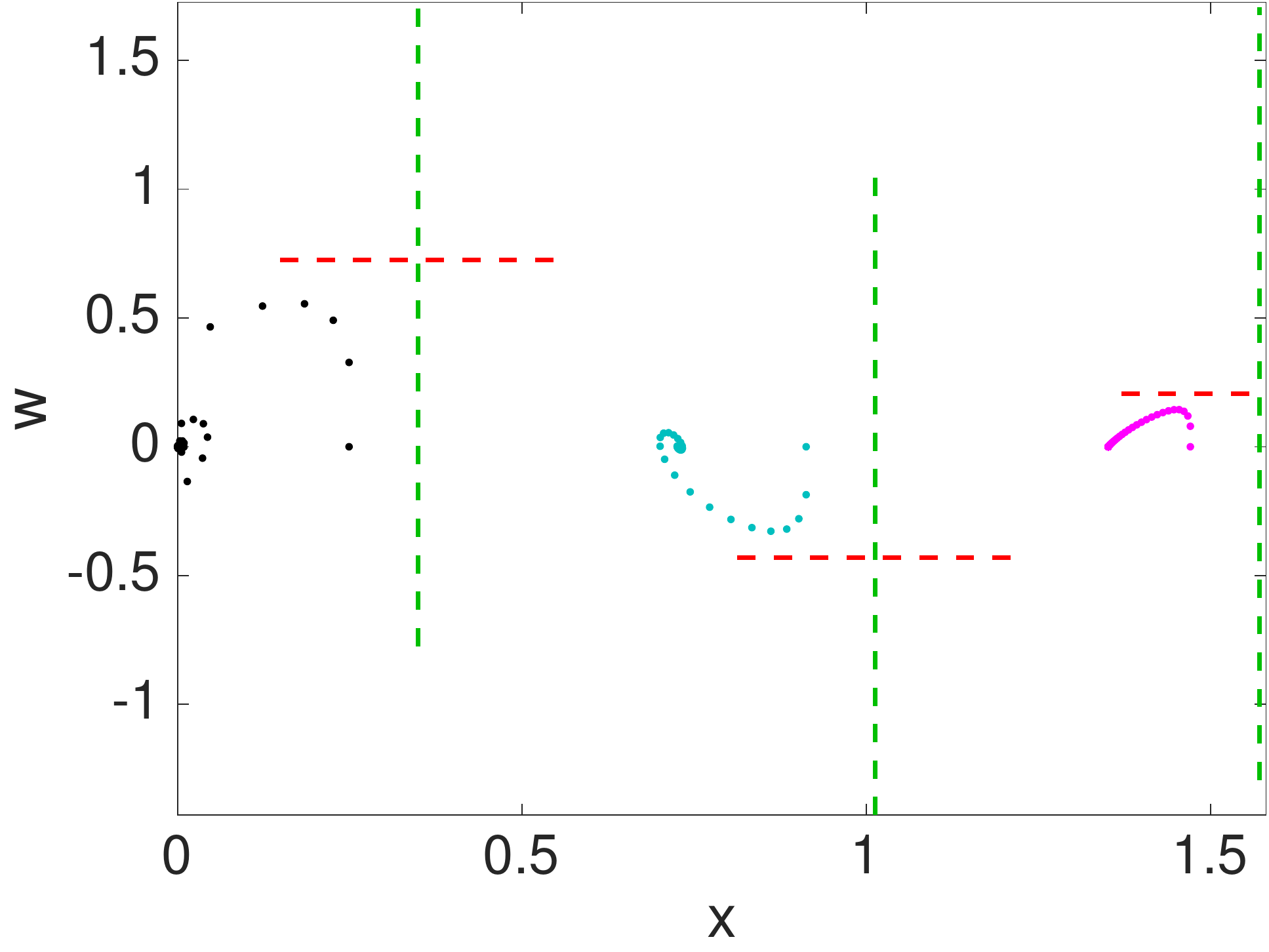}
\caption{Stable foci: $\mu = .5$}
\label{Fig: mu5}
\end{subfigure}
\hfill
\begin{subfigure}[t]{0.49\textwidth}
\includegraphics[width=\textwidth]{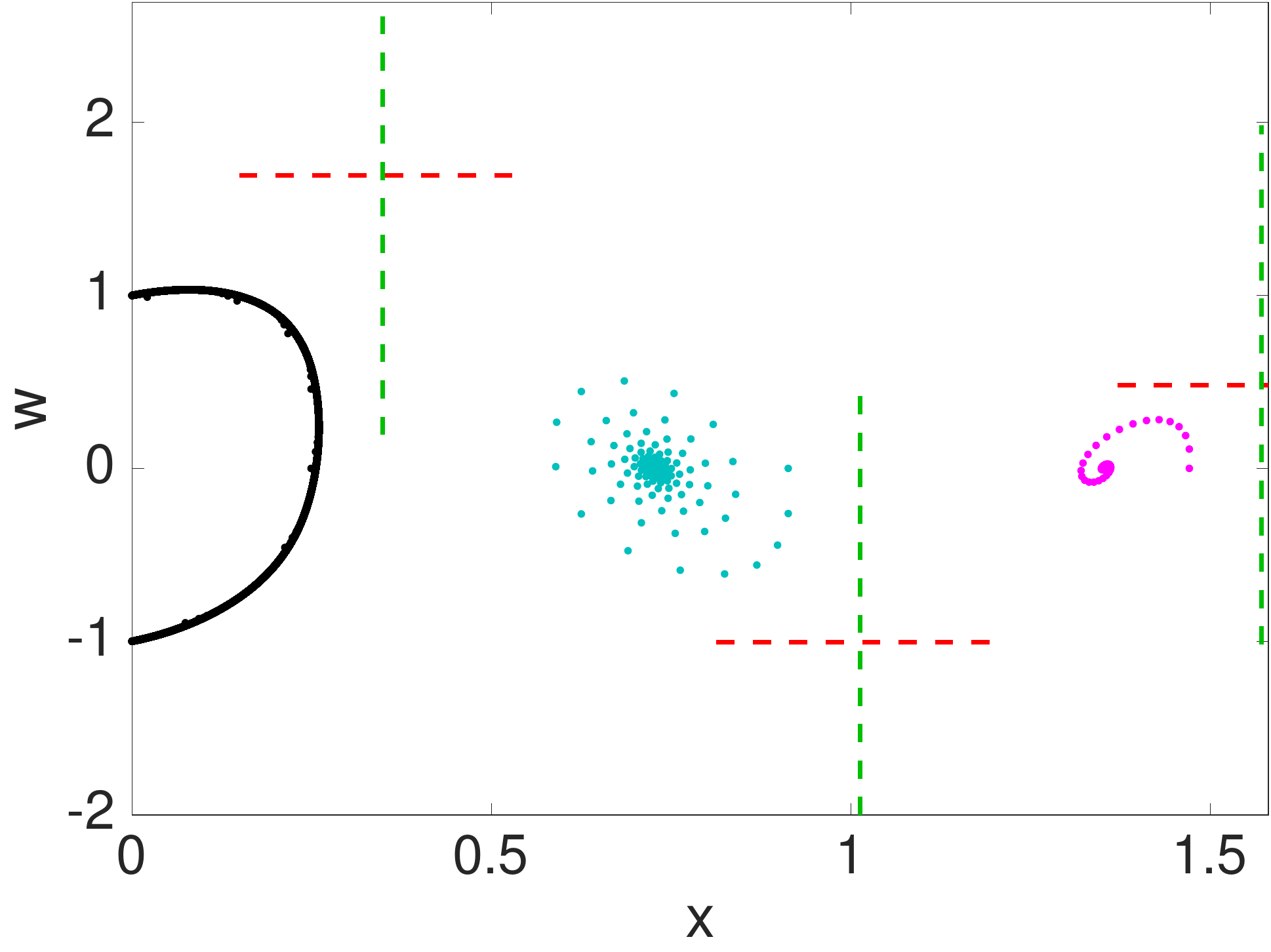}
\caption{Invariant circle about \textbf{origin}: $\mu = .7$}
\label{Fig: mu7}
\end{subfigure}
\begin{subfigure}[t]{0.49\textwidth}
\includegraphics[width=\textwidth]{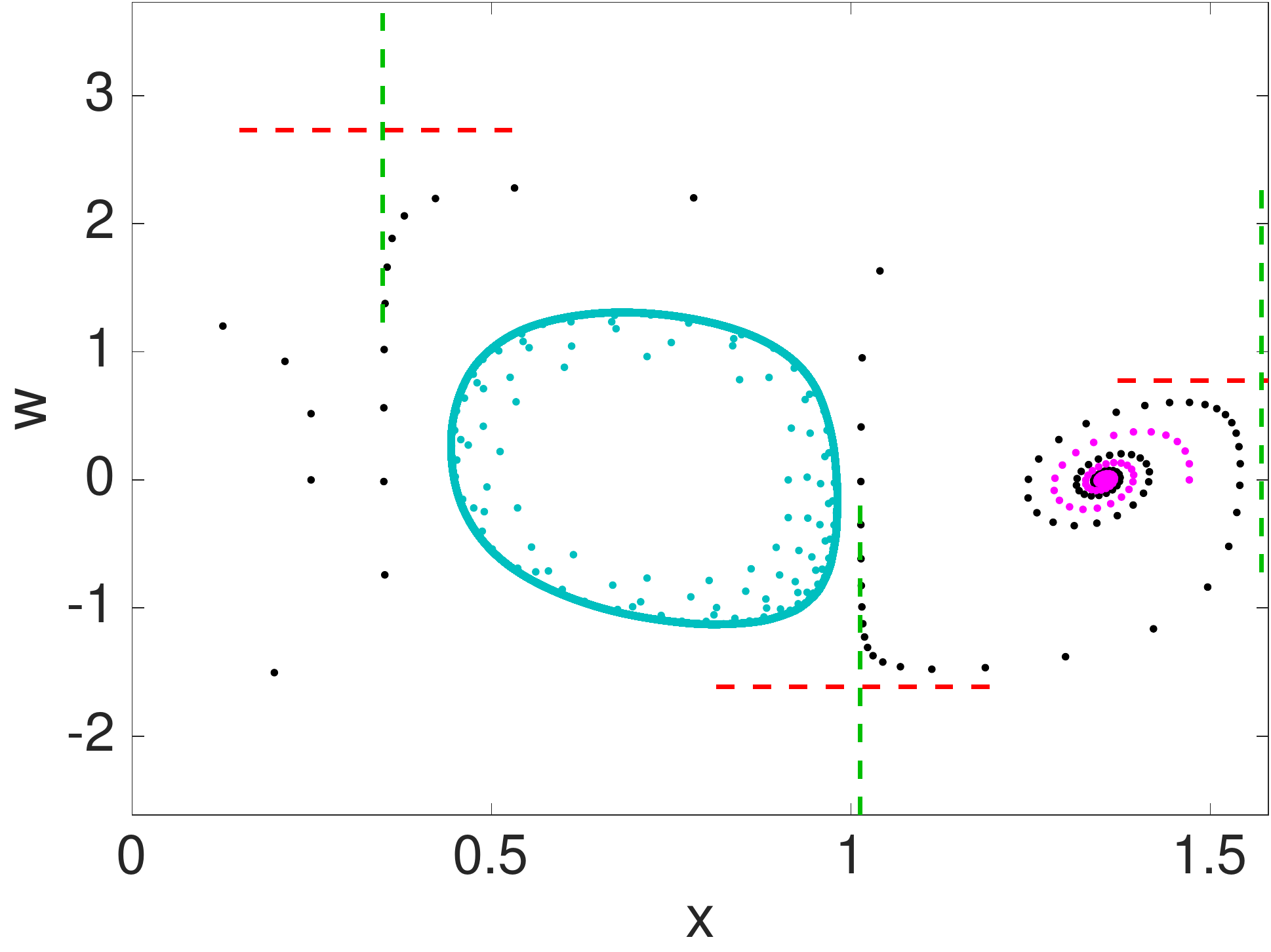}
\caption{Invariant circle about {\color{cyan}$(\approx .7269,0)$}: $\mu = .79$}
\label{Fig: mu79}
\end{subfigure}
\hfill
\begin{subfigure}[t]{0.49\textwidth}
\includegraphics[width=\textwidth]{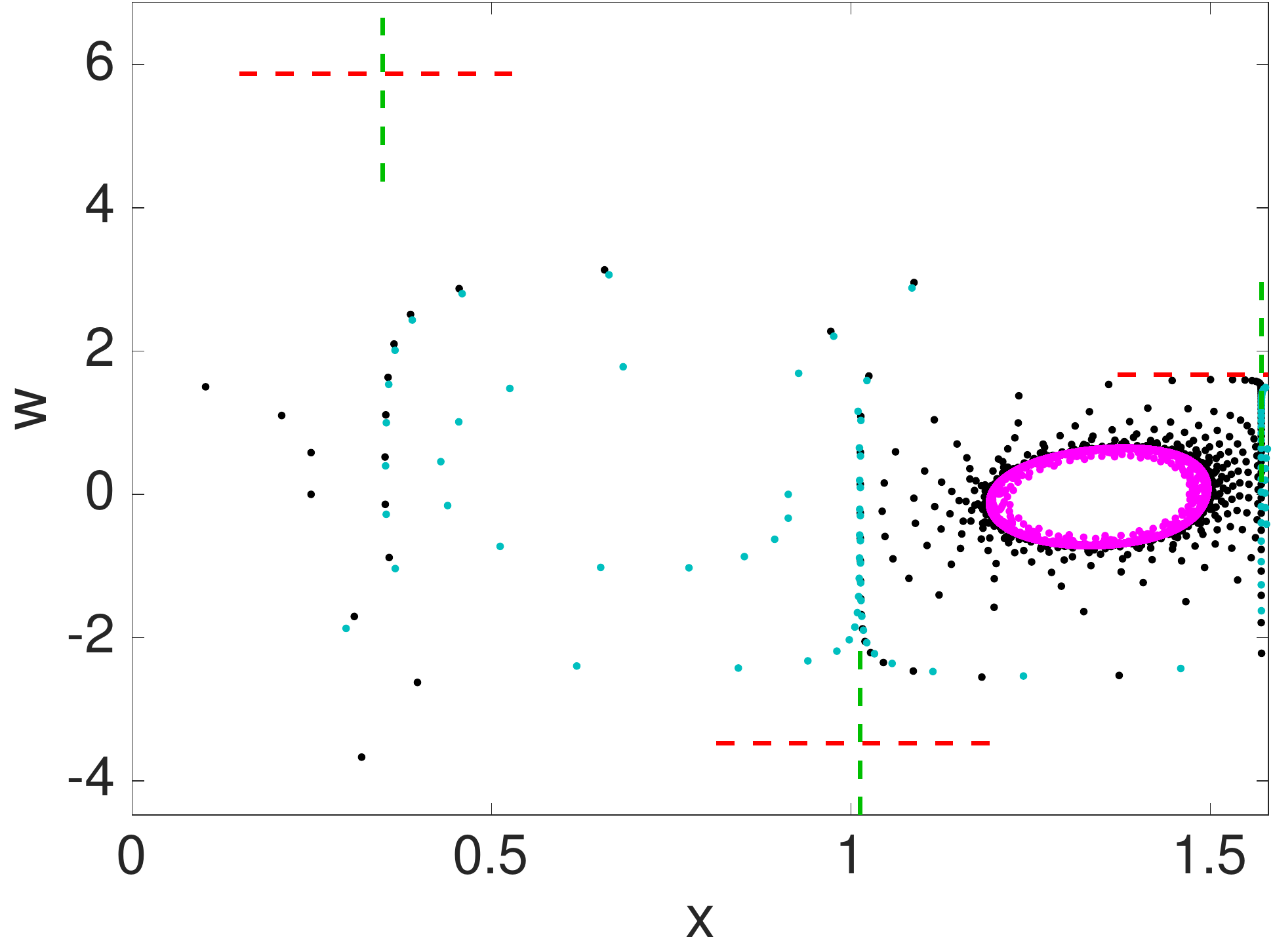}
\caption{Invariant circle about {\color{magenta}$(\approx 1.3515,0)$}: $\mu = .89$}
\label{Fig: mu89}
\end{subfigure}
\caption{For each plot the {\color{green}green} lines represent the linear
{\color{green}stable manifold} and the {\color{red}red} lines represent the linear {\color{red}unstable manifold} at 
the respective saddle fixed points.  The \textbf{black} markers represent iterates originating from a neighborhood 
of the \textbf{origin}, the {\color{cyan}cyan} markers represent iterates originating from a neighborhood of 
{\color{cyan}$(\approx .7269,0)$}, and the {\color{magenta}magenta} markers represent the iterates originating 
from a neighborhood of  {\color{magenta}$(\approx 1.3515,0)$}.}
\label{Fig: Supercritical}
\end{figure}

At the fixed points a supercritical N--S bifurcation occurs for $\mu_* = 0.64894$, $0.742027$, $0.879451$, 
respectively, for which $\d = -4076.61$, $-1747.52$, $-410.779$.  We also have three other fixed points:
$(\approx 0.3490,\approx\frac{\mu}{1-\mu}\Psi(0.3490))$, $(\approx 1.0128,\approx\frac{\mu}{1-\mu}\Psi(1.0128))$, $(\pi/2,\frac{\mu}{1-\mu}\Psi(\pi/2))$, which are saddles (the $w$ position varying with
$\mu$) and therefore not subject to N--S bifurcations.  Let us vary the parameter $\mu$ from $\mu = 0.5$ to
$\mu = 0.89$. We illustrate the progression of the bifurcations in Fig. \ref{Fig: Supercritical}. In Fig.
\ref{Fig: mu5}, the relevant fixed points are all stable foci.  Next, in Fig. \ref{Fig: mu7}, we have passed the
critical value for the origin and a stable invariant circle is now visible. Finally, in Fig. \ref{Fig: mu79} and
\ref{Fig: mu89}, the critical values for the next two fixed points, respectively, are passed.  We notice in Fig. 
\ref{Fig: mu79} and \ref{Fig: mu89} that the preceding fixed point(s), respectively, are now unstable focus(foci).  
This shows that each fixed point undergoes a supercritical N--S bifurcation in accordance with Theorem
\ref{Thm: thetheorem}.  It should be noted that since the test functions are periodic, on the entire domain of
$x \in \R$, N--S bifurcations occur at (countably) infinitely many fixed points
with the final bifurcation occurring for fixed points corresponding to $x_* \approx 1.3515$.

Physically, the droplet about any fixed point is initially (Fig. \ref{Fig: mu5}) oscillating with decaying radius. Then, 
due to the nature of the single eigenmode used by Gilet, the droplet corresponding to a position near the origin 
will undergo a N--S bifurcation (Fig. \ref{Fig: mu7}). This corresponds to the droplet oscillating about the origin 
with a fixed radius.  Then as $\mu$ is varied further, the invariant circle about the origin disintegrates and it 
becomes unstable, which corresponds to the orbit of the droplet increasing radially until it reaches the 
neighborhood of one of the neighboring fixed points.  Now, since a N--S bifurcation has already occurred
at the next fixed point (Fig. \ref{Fig: mu79}), the iterates will bypass this region and converge to the
following fixed point, which corresponds to the droplet jumping 
to an orbit around a different center. Similar behavior occurs at every other N--S fixed point.

In addition to N--S bifurcations, we notice a curious phenomenon.  As the saddle fixed points move away from the  
N--S fixed points the trajectories from the neighborhood of one  N--S fixed point crosses into that of another.  
Furthermore, as the invariant circle increases in radius it seems to collide (more precisely, gets arbitrarily close) 
with a stable manifold of one of its neighboring saddles just before the onset of chaos.
We shall analyze this further in Section \ref{Sec: Homoclinic}.

\subsubsection{Subcritical Neimark--Sacker bifurcations for $\beta = 5\pi/6$}

For $\beta = 5\pi/6$ let us once again consider the map (\ref{Eq: theothermap}), however let us simplify matters 
and restrict our domain to a neighborhood of the origin.  Here lies a subcritical N--S bifurcation at
$\mu_* = 0.999847$ where $\d = 4.88756$. Since this is subcritical, let us vary our parameter $\mu$ 
backwards from $\mu = 0.9999$ to $\mu = .999$.  This is illustrated in Fig. \ref{Fig: Subcritical}.  First the origin 
is an unstable focus, then as $\mu$ passes $\mu_*$ backwards the origin becomes a stable focus, however if an 
initial point is taken further out the iterates diverge, which indicates an unstable invariant circle.  In Fig.
\ref{Fig: beta2mu999}, we represent this unstable invariant circle by the black dashed curve.

\begin{figure}[htbp]
\centering
\begin{subfigure}[t]{0.48\textwidth}
\includegraphics[width=\textwidth]{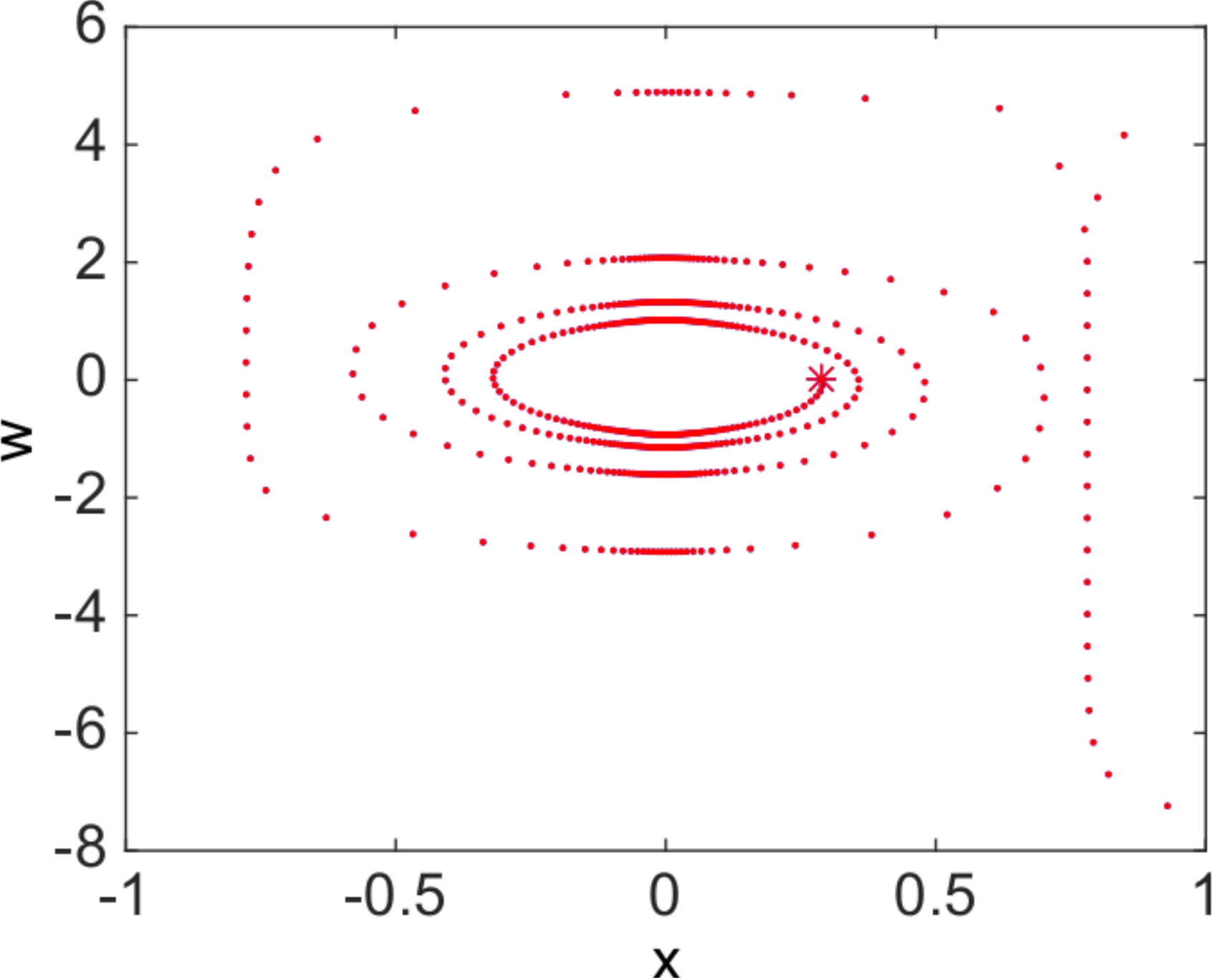}
\caption{Unstable focus: $\mu = 0.9999$}
\label{Fig: beta2mu9999}
\end{subfigure}
\begin{subfigure}[t]{0.495\textwidth}
\includegraphics[width=\textwidth]{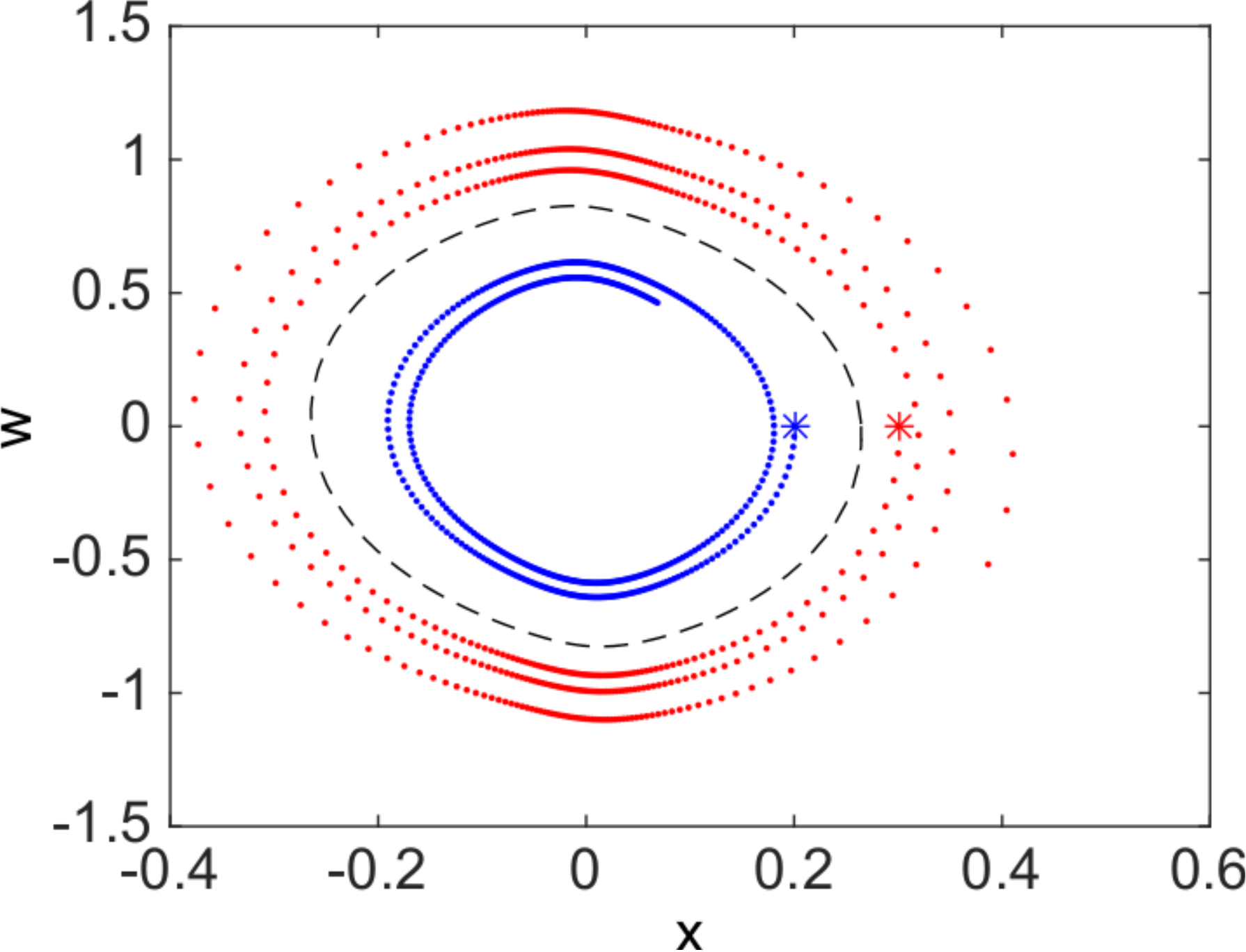}
\caption{Unstable invariant circle: $\mu = 0.999$}
\label{Fig: beta2mu999}
\end{subfigure}
\caption{A star denotes the initial point.  In the second plot the {\color{blue}blue}
iterates have an initial point {\color{blue}inside} the invariant circle, and the {\color{red}red} iterates have an 
initial point {\color{red}outside} the invariant circle.  Finally, the black dashed curve represents the unstable 
invariant circle.}
\label{Fig: Subcritical}
\end{figure}

While from a dynamical systems perspective, we vary the parameter backward for a subcritical bifurcation, this 
would not make sense for physical intuition.  Rather let us begin at Fig. \ref{Fig: beta2mu999}, where a droplet 
close enough to the origin would oscillate with a decaying radius, however if the droplet were exactly on the 
unstable invariant circle it would remain at that constant radius, and if the droplet were slightly further out it 
would oscillate with a growing radius and escape the neighborhood of the origin.  Now as we increase the memory 
the invariant circle shrinks until the critical parameter value $\mu_* = 0.999847$, when it ceases to exist and the 
origin becomes unstable, which corresponds to the droplet anywhere near the origin to oscillate with increasing 
radius.  It should be noted that the growth of the radius is not indefinite,
and at some point it may be drawn to another fixed 
point or invariant circle, which is usually the case.

\subsection{Simulation results for Neimark--Sacker bifurcation in $C$}

Varying $C$ produces quite exotic trajectories, one example of which we present in this section, holding
$\mu = 0.5$ and $\beta = \pi/3$.  Consider two fixed points, $(x_*,0)$ such that 
\begin{align*}
x_* &= -\tan^{-1}\left(\sqrt{\frac{15+\sqrt{3}+\sqrt{6(8-\sqrt{3})}}{9-\sqrt{3}-\sqrt{6(8-\sqrt{3})}}}\right)
\approx -1.3515,\\
x_* &= -\pi + \tan^{-1}\left(\sqrt{\frac{15+\sqrt{3}-\sqrt{6(8-\sqrt{3})}}{9-\sqrt{3}
+\sqrt{6(8-\sqrt{3})}}}\right)
\approx -1.7901.
\end{align*}
At these fixed points a supercritical Neimark--Sacker bifurcations occurs simultaneously due to the symmetry for 
$C_* = 0.36477$, where $\d = -578.487$.  The other fixed point is $(-\pi/2,\Psi(-\pi/2))$, which is a saddle.

For this case, varying $C$ forwards corresponds to the drop size increasing, which leads to greater inertia.  With 
increased inertia, the distance the droplet travels due to an impact is greater, and hence sensitive dependence 
may come into play. First, the N--S bifurcation is illustrated in Fig. \ref{Fig: NS_C}. Here the physical interpretation 
of the bifurcation is the same as in section \ref{Sec: Bif_mu}, the only difference being the change in the size of 
the droplet as opposed to changes in memory.

\begin{figure}[htbp]
\centering
\begin{subfigure}[t]{0.49\textwidth}
\includegraphics[width=\textwidth]{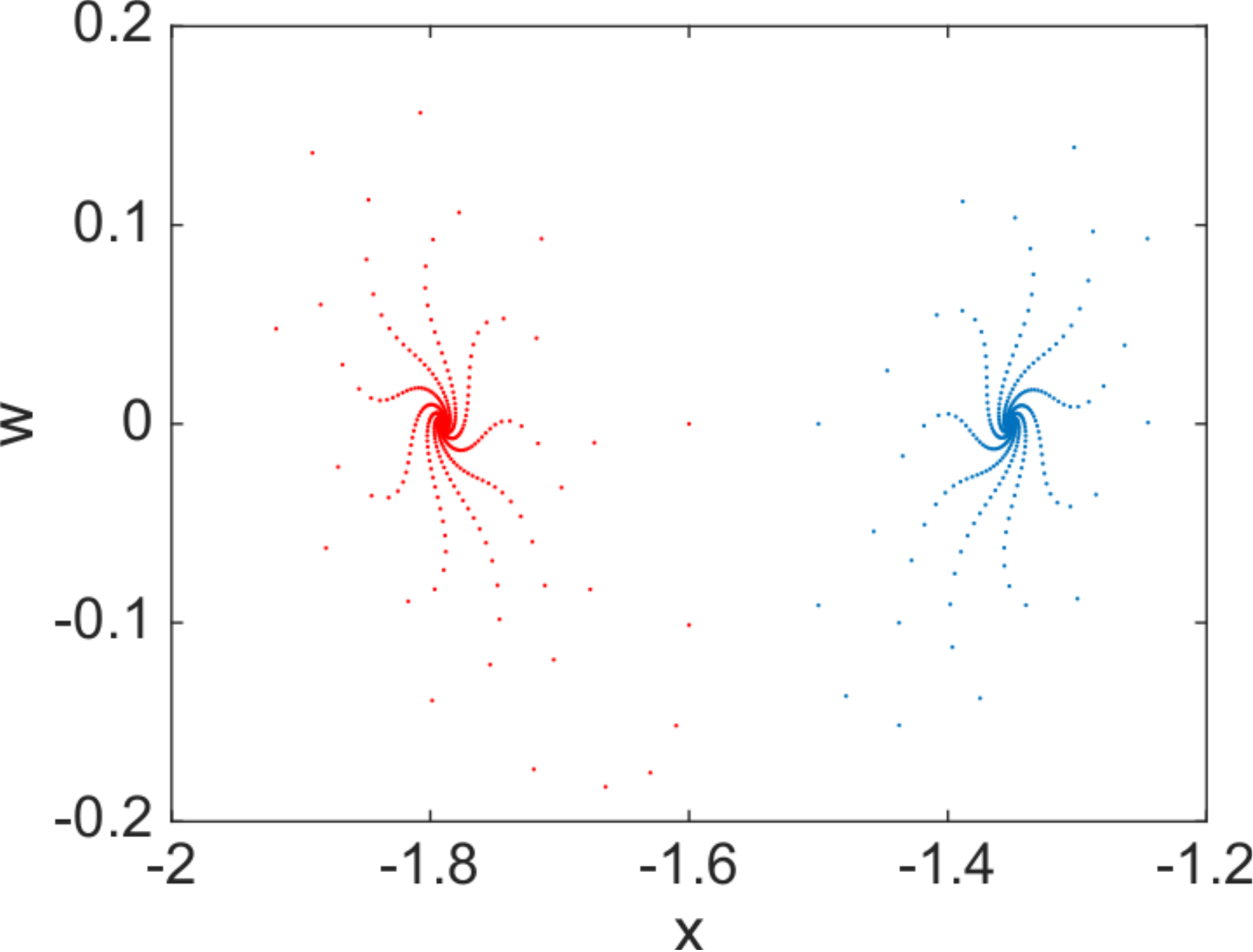}
\caption{Stable foci: $C = 0.35$}
\label{Fig: C35}
\end{subfigure}
\begin{subfigure}[t]{0.49\textwidth}
\includegraphics[width=\textwidth]{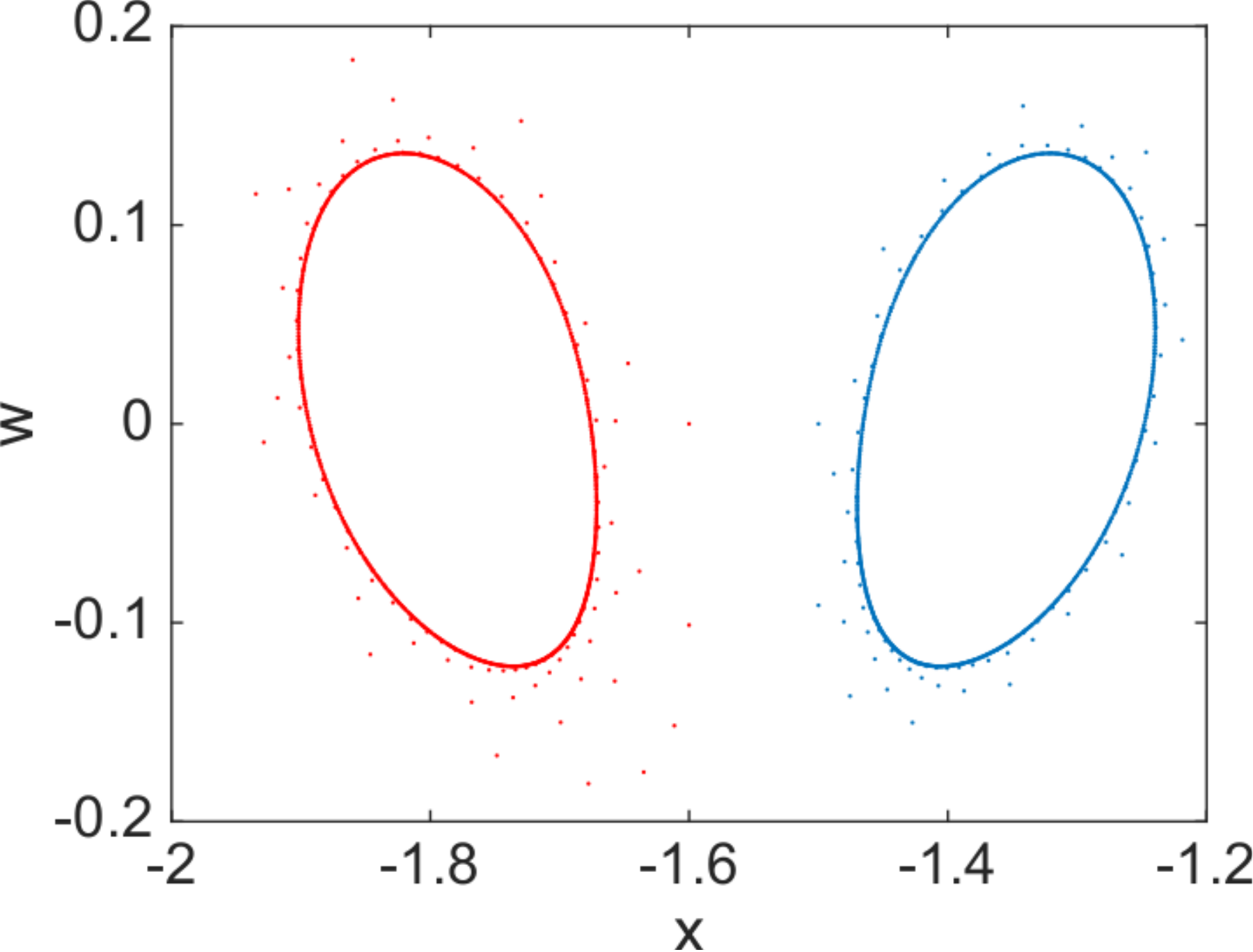}
\caption{Stable invariant circle: $C=0.4$}
\label{Fig: C4}
\end{subfigure}
\caption{The {\color{red} red} trajectories start on the left of the saddle
{\color{blue} blue} trajectories start on the right of the saddle.}
\label{Fig: NS_C}
\end{figure}

Now, as we increase $C$, we obtain the exotic trajectories as illustrated in Fig. \ref{Fig: Exotic}.  We observe the 
characteristic stretching and folding of a chaotic system, which may lead to fractal structures. The way in which 
the invariant circles approach the stable manifold that separates them shows evidence of what appears to be an 
unusual and possibly new type of global bifurcation, which is briefly discussed further in the next section.

\begin{figure}[htbp]
\centering
\begin{subfigure}[t]{0.49\textwidth}
\includegraphics[width=\textwidth]{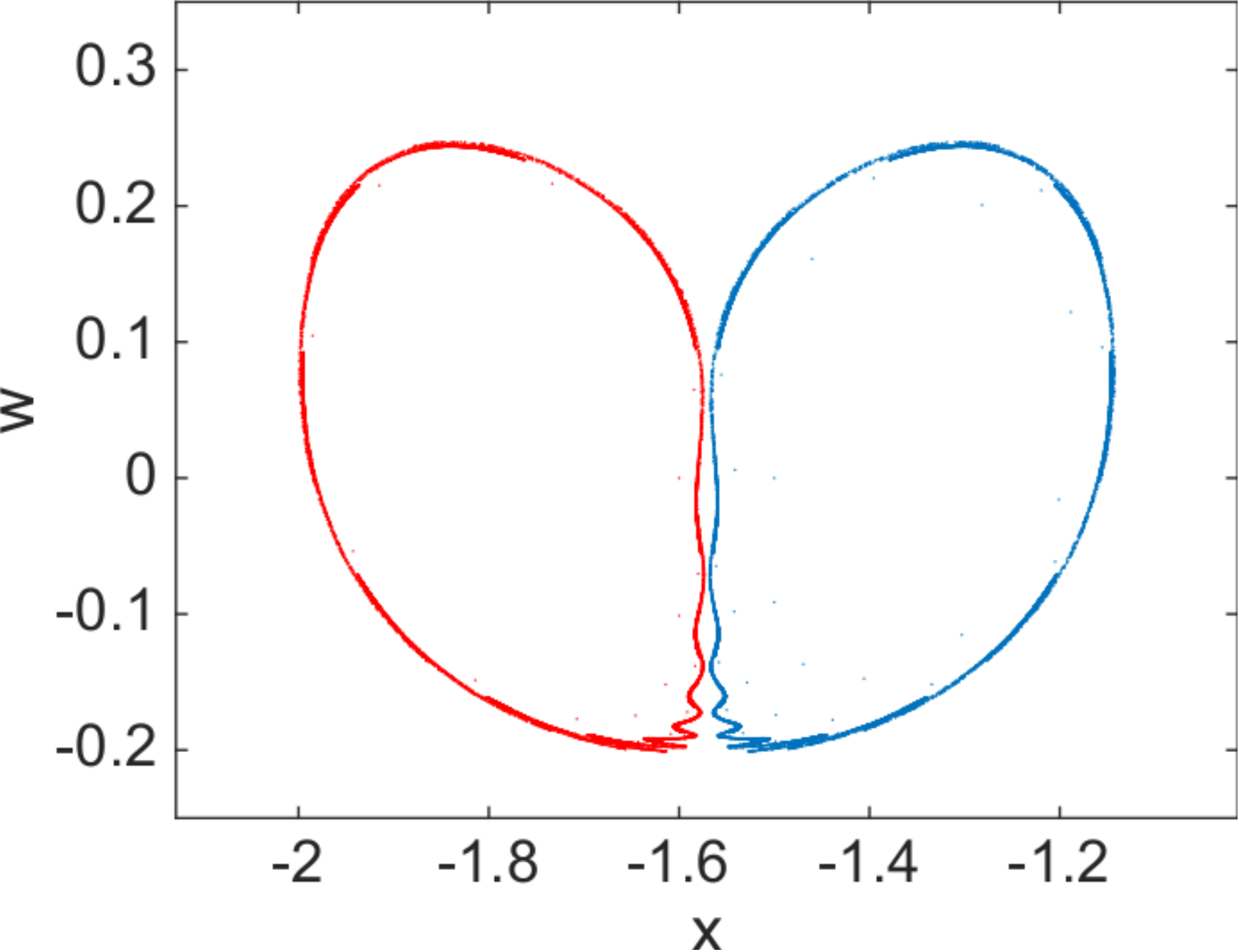}
\caption{Deformed invariant circles: $C = 0.5$}
\label{Fig: C5}
\end{subfigure}
\begin{subfigure}[t]{0.49\textwidth}
\includegraphics[width=\textwidth]{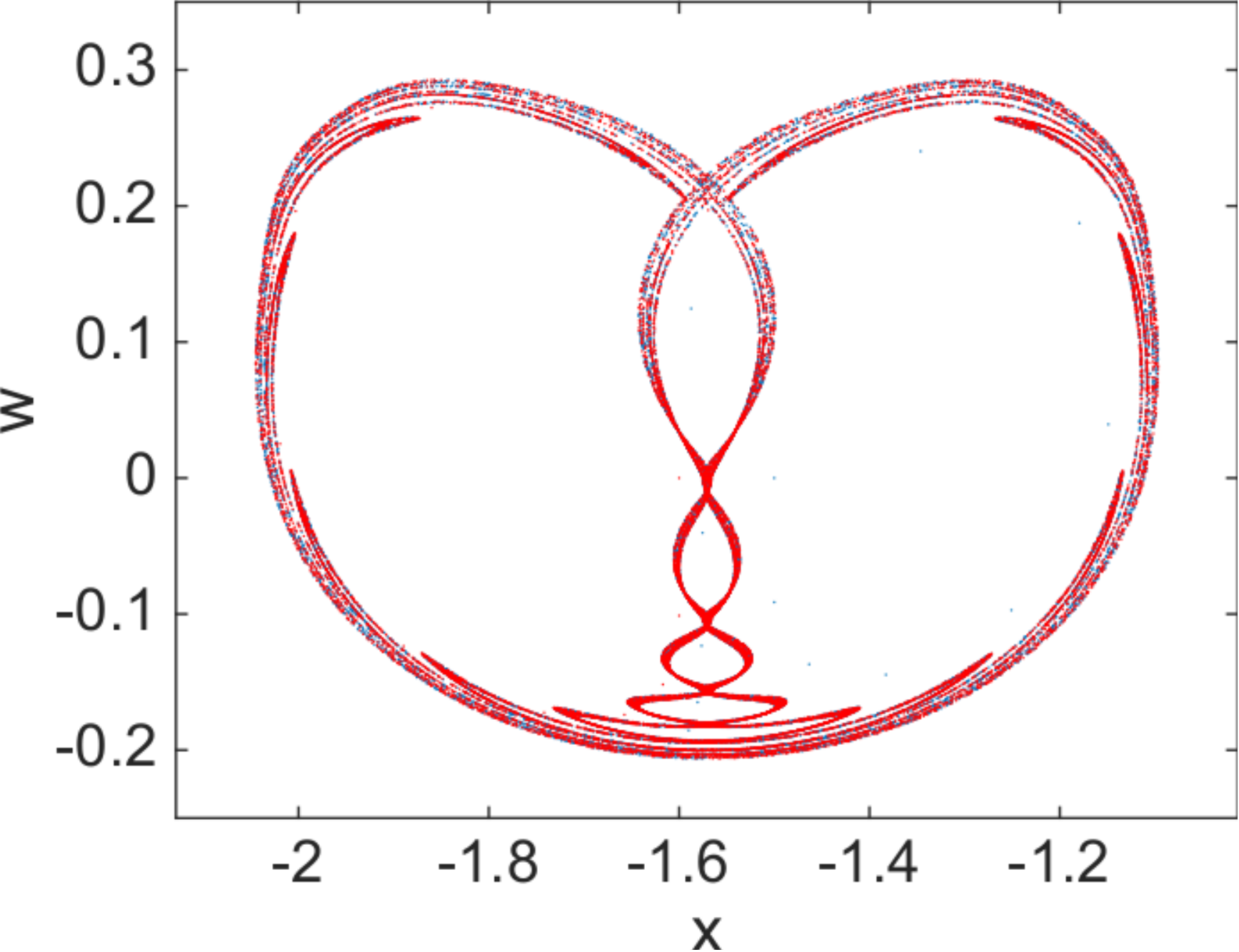}
\caption{Evidence of global bifurcation: $C=0.6$}
\label{Fig: C6}
\end{subfigure}
\caption{The {\color{red} red} trajectories start on the left of the saddle
{\color{blue} blue} trajectories start on the right of the saddle.}
\label{Fig: Exotic}
\end{figure}

\section{Evidence of global bifurcations leading to chaos}\label{Sec: Homoclinic}

In Fig. \ref{Fig: mu89} we observe the appearance of quite exotic dynamics.  If we continue to vary $\mu$ we find 
more exotic, chaotic-like, dynamics.  We illustrate this in Fig. \ref{Fig: Chaos} where we see a scatter of iterates in 
a manner indicative of chaos. Moreover, we observe as $\mu$ is varied, the iterates remain within a compact set.  
Here the iterates seem to satisfy the transitivity and sensitivity conditions for chaos. So far, all we have are 
simulations and informed intuition: further investigation is clearly required in order to obtain a more precise 
characterization these dynamical properties.

\begin{figure}[htbp]
\centering
\begin{subfigure}[t]{0.48\textwidth}
\includegraphics[width=\textwidth]{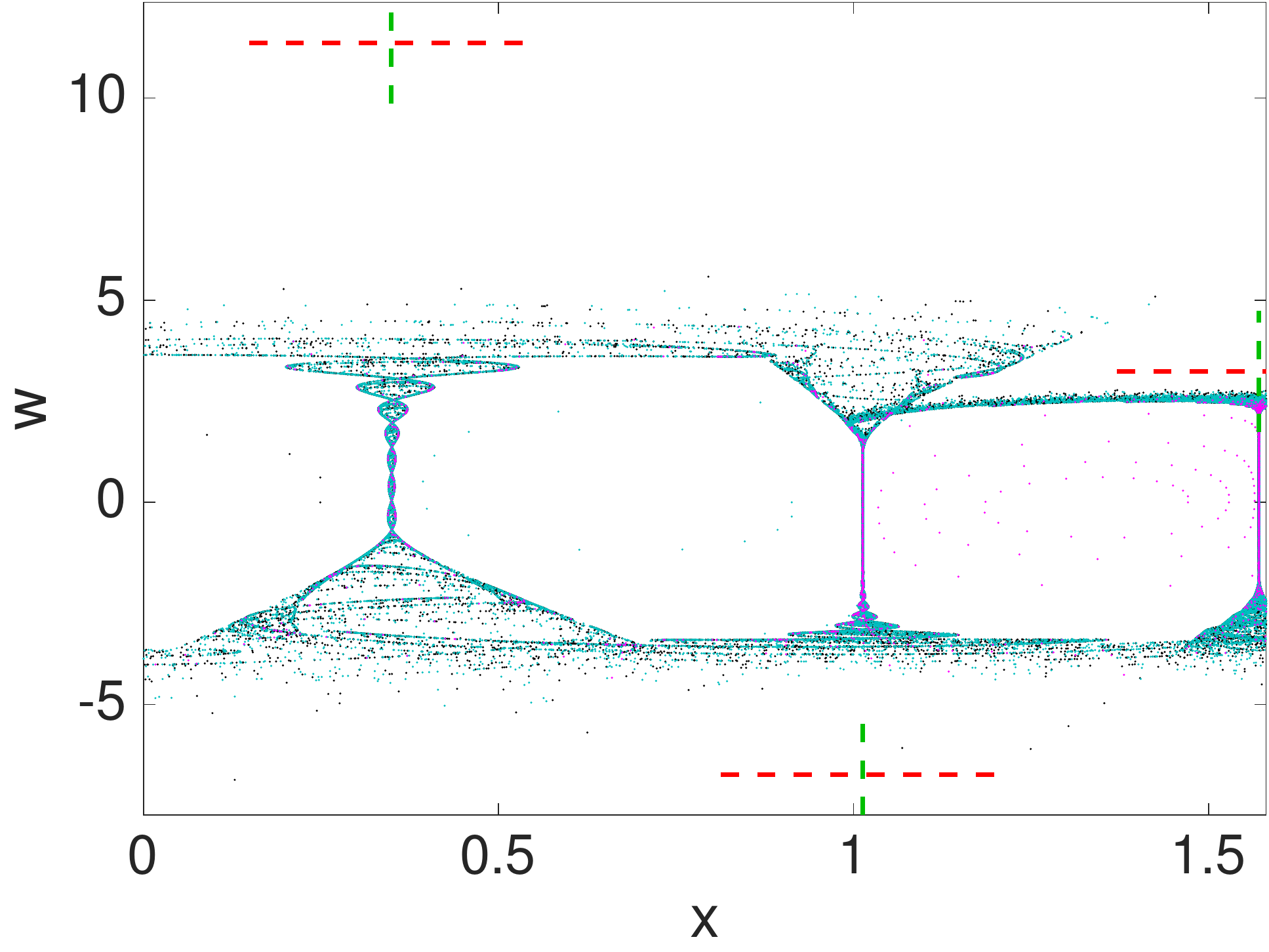}
\caption{Evidence of chaos: $\mu = 0.94$}
\label{Fig: mu94}
\end{subfigure}
\begin{subfigure}[t]{0.49\textwidth}
\includegraphics[width=\textwidth]{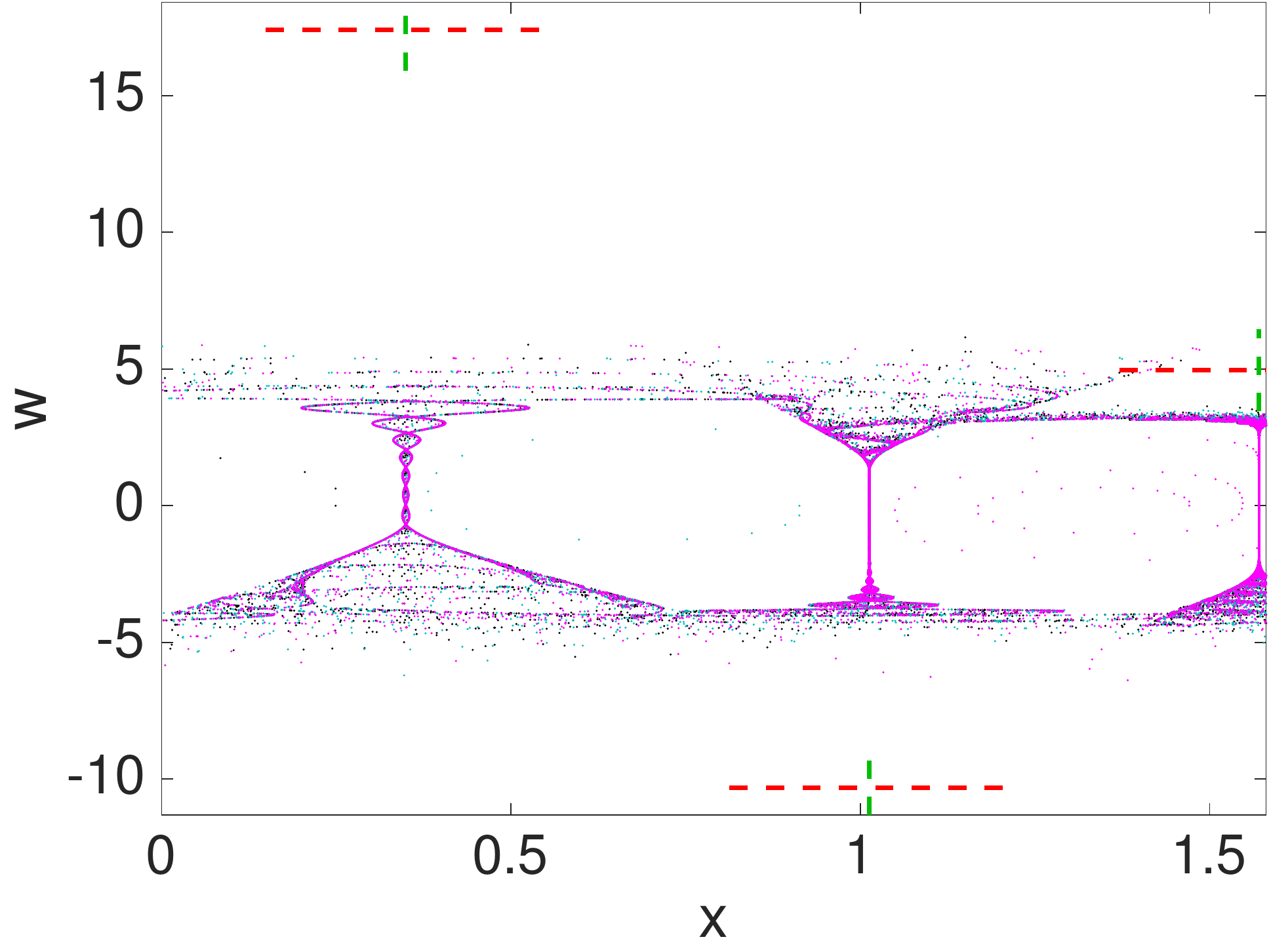}
\caption{Evidence of chaos: $\mu = 0.96$}
\label{Fig: mu96}
\end{subfigure}
\caption{Exotic, chaotic-like, trajectories.  For each plot the {\color{green}green}
lines represent the linear {\color{green}stable manifold} and the {\color{red}red} lines represent the linear 
{\color{red}unstable manifold} at the respective saddle fixed points.  The \textbf{black} markers represent 
iterates originating from a neighborhood of the \textbf{origin}, the {\color{cyan}cyan} markers represent iterates 
originating from a neighborhood of {\color{cyan}$(\approx .7269,0)$}, and the {\color{magenta}magenta} 
markers represent the iterates originating from a neighborhood of  {\color{magenta}$(\approx 1.3515,0)$}.}
\label{Fig: Chaos}
\end{figure}

Now let us vary $\mu$ near the onset of chaos to study the bifurcation. This is illustrated in Fig.
\ref{Fig: Homoclinic}.  First (Fig. \ref{Fig: mu913}) we observe our invariant circles for $\mu = 0.913$ as before, 
then as we increase to $\mu = 0.914$ the invariant circles start to disintegrate (Fig. \ref{Fig: mu914}).  Finally, 
in Fig. \ref{Fig: mu915}, for $\mu = 0.915$ we see the iterates bouncing between the neighborhoods of their 
respective N--S fixed point and perhaps nested invariant circles (both stable and unstable).  This provides evidence 
of an exotic new bifurcation as the invariant circle approaches the stable manifold of the saddle and perhaps even 
a cascading N--S doubling bifurcations, which is an infinite succession of stability changes of invariant circles that 
create new pairs of invariant circles (see, e.g. \cite{BRS2009}, where they are called cascading Hopf doubling 
bifurcations).

\begin{figure}[htbp]
\centering
\begin{subfigure}[t]{0.32\textwidth}
\includegraphics[width=\textwidth]{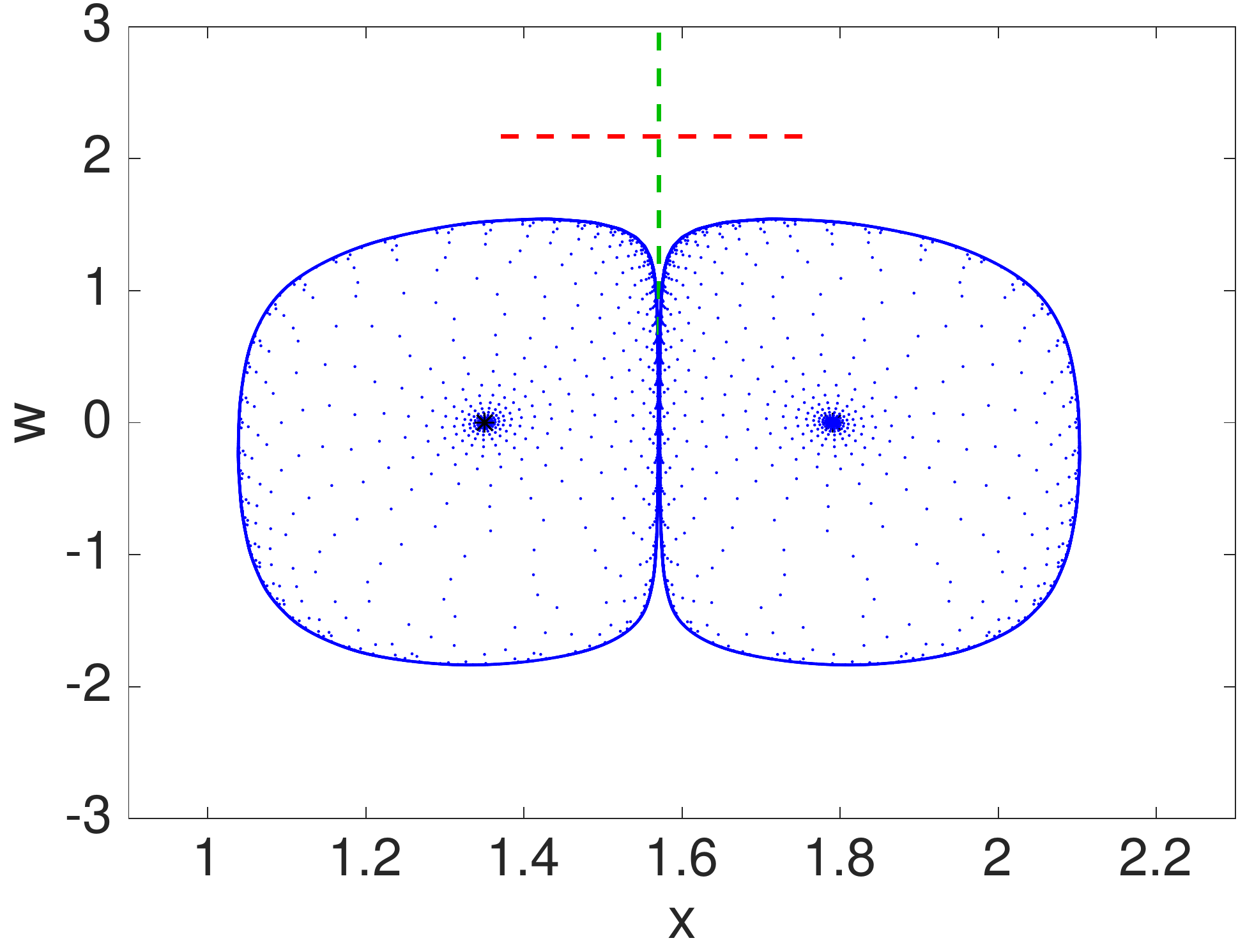}
\caption{Invariant circle: $\mu = 0.913$}
\label{Fig: mu913}
\end{subfigure}
\begin{subfigure}[t]{0.32\textwidth}
\includegraphics[width=\textwidth]{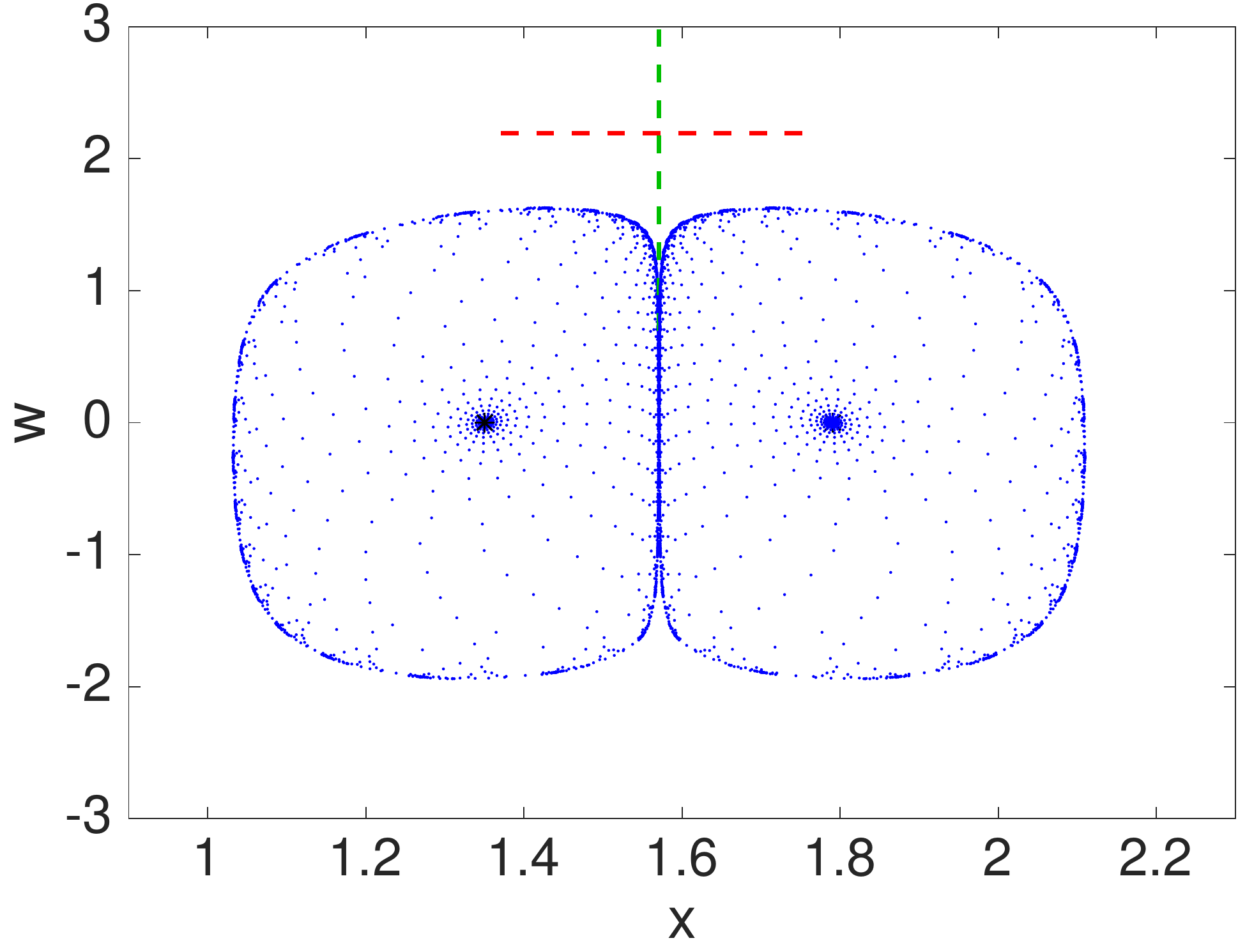}
\caption{Invariant circle disintegrating: $\mu = 0.914$}
\label{Fig: mu914}
\end{subfigure}
\begin{subfigure}[t]{0.32\textwidth}
\includegraphics[width=\textwidth]{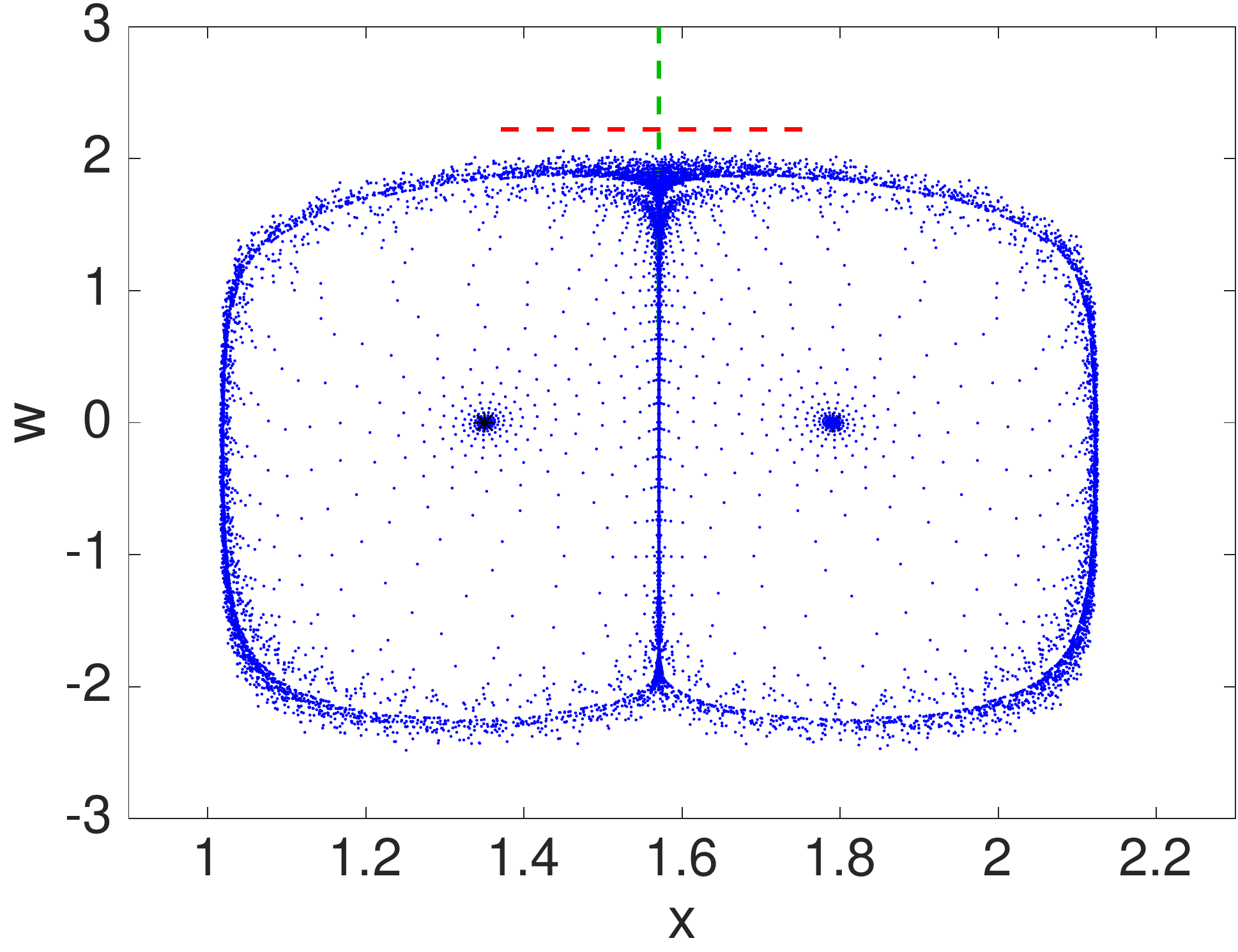}
\caption{Iterates after collision: $\mu = 0.915$}
\label{Fig: mu915}
\end{subfigure}
\caption{The {\color{green}green} lines represent the linear
{\color{green}stable manifold} and the {\color{red}red} lines represent the linear {\color{red}unstable manifold} at 
the respective saddle fixed points.  The {\color{blue}blue} markers are the {\color{blue}iterates}.}
\label{Fig: Homoclinic}
\end{figure}

Furthermore, we notice in our figures that iterates prefer to intersect the linear stable manifolds. If we do not 
restrict the region of the plot, we observe that the iterates seem to prefer certain positions as shown in Fig. 
\ref{Fig: Quantization}. In \cite{Gilet14}, Gilet gives a detailed study of the statistics associated with the positions 
of the iterates and in the simulations he observes a probability density function inversely proportional to
$|\Psi'(x)|$.  While we do not 
study the statistics, we observe that the linear stable manifolds pass through the roots of $\Psi'(x)$, which hints 
at connections between the dynamical systems aspects and statistical analysis of this model.

\begin{figure}[htbp]
\centering
\includegraphics[width = 0.75\textwidth]{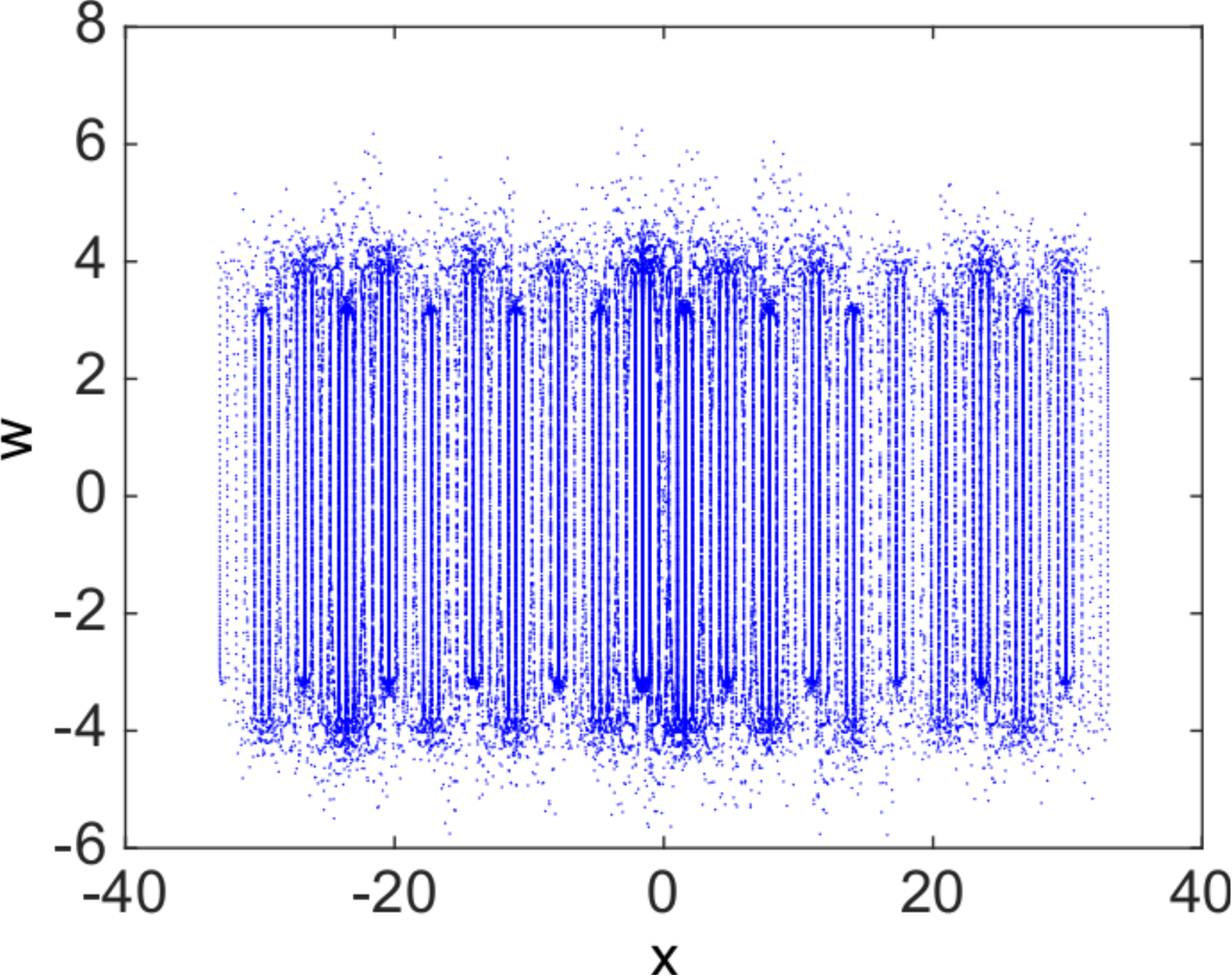}
\caption{Iterates of the map (\ref{Eq: theothermap}) generated from four initial points chosen near
the origin.}
\label{Fig: Quantization}
\end{figure}

\section{Conclusions}

Previous investigations of the walking phenomena have shown evidence of exotic dynamics, such as Hopf 
bifurcations and chaos, mainly through experiments and numerical simulations. However, the equations have 
generally been too complex to prove the existence of certain bifurcations, chaos, and other topological properties.  
In recent years, discrete models have been studied by Fort \ea\cite{FEBMC10}, Eddi \ea\cite{ESMFRC11}, 
Shirokoff \cite{Shirokoff13}, and Gilet \cite{Gilet14}. The models of Shirokoff \cite{Shirokoff13} and Gilet
\cite{Gilet14} are discrete dynamical systems, which are often more accessible to rigorous analysis. In this paper 
we prove some of the properties observed and conjectured for Gilet's model; namely, the existence of N--S 
bifurcations. Similar dynamical systems analysis is likely in future to prove additional results concerning 
bifurcations, chaos, structural stability, and other properties. Proving these properties is not only important 
mathematically, but may also provide insight into the phenomena observed for walkers and for their continuous 
models.

We proved Gilet's conjecture \cite{Gilet14} about the existence of supercritical N--S bifurcations when varying the 
parameter $\mu$. In doing so, we are also able to prove the existence of subcritical N--S bifurcations in the 
parameter $\mu$, which had been missed by the original simulations. We then varied the parameter $C$, which 
had not been studied before, and proved the existence of N--S bifurcations there as well.

In order to verify the validity of our theory, the model was simulated using test functions for $\Psi$.  For the sake 
of consistency the test functions used were of the same form as \cite{Gilet14}, but the shapes were changed by 
varying $\beta$. By doing so were able to analyze the transition between supercritical and subcritical N--S 
bifurcations (Fig. \ref{Fig: d}) for both $\mu$ and $C$. We reproduced one of the simulations done in 
\cite{Gilet14} in order to show consistency between the two works.  Then we ran additional simulations of the 
various cases in our theorems, all of which demonstrate complete agreement between our results and the 
numerics.

In addition, we observed evidence of more exotic phenomena, which seem to lead to chaos.  One such curious 
phenomenon is a homoclinic-like bifurcation when the invariant circles get ever closer to the stable manifold of a 
neighboring saddle fixed point.  This also appears to give rise to other invariant circles leading to cascading N--S 
bifurcations wherein the invariant circles change stability and give birth to twin stable invariant circles.

Finally, while not the focus of this work, Gilet, as did we, observed interesting statistics for the position of the 
iterates.  We notice that after the onset of chaos the iterates prefer certain positions, namely along the stable 
manifolds of each saddle fixed point.  In \cite{Gilet14}, Gilet studies the statistics in detail through the use of 
approximate probability density functions. It may be possible to obtain a more analytical characterization using 
measurable dynamics theory, which is a problem worthy of future investigation.  He writes, \textquotedblleft 
Future work could include the search for wave-particle coupling dynamics that yields
$\text{PDF}(x) \sim |\Psi(x)|^2$ as observed in quantum mechanics.\textquotedblright
We agree that this would 
be a fruitful endeavor and perhaps studying the dynamics of the system further will provide insight into the right 
choice of wave-particle coupling.

\section{Acknowledgement}

The authors would like to thank John Bush, Ivan Christov, and David Shirokoff for fruitful discussions and feedback. 
And special thanks are due Anand Oza, Tristan Gilet, and Roy Goodman for lengthy discussions and exchange of 
countless ideas. Finally, thanks are due the referees for their insightful constructive criticism that led to this 
improvement of the original manuscript.

\bibliographystyle{unsrt}
\bibliography{Bouncing_droplets}
\nocite{*}

\end{document}